\documentclass[11pt]{article}
\usepackage{booktabs}
\usepackage{amsfonts}
\usepackage{amsmath}
\usepackage{graphicx}
\usepackage{amsthm}
\usepackage{amssymb}
\usepackage{multirow}
\usepackage{makecell}
\usepackage[vlined,linesnumbered,ruled,resetcount]{algorithm2e}
\usepackage[round]{natbib}
\usepackage[colorlinks,linkcolor=magenta,filecolor=blue,citecolor=blue,urlcolor=blue]{hyperref}%
\usepackage[top=1in, left=1in, right=1in, bottom=1in]{geometry}
\usepackage{subfigure}
\usepackage{amssymb}
\usepackage{pifont}
\usepackage{mathtools}
\usepackage{stmaryrd}
\usepackage[T1]{fontenc}

\usepackage{enumitem}
\usepackage{algorithmic}
\newtheorem{theorem}{Theorem}
\newtheorem{definition}{Definition}

\newtheorem{problem}{Problem}

\begin{document}

\title{\bf\huge Proximity Graph Maintenance for Fast Online  Nearest Neighbor Search}

\author{\vspace{0.3in}\\ Zhaozhuo Xu, Weijie Zhao,  Shulong Tan,
 Zhixin Zhou, Ping Li\\\\
Cognitive Computing Lab\\
Baidu Research\\
10900 NE 8th St. Bellevue, WA 98004, USA\\\\
\texttt{\{zhaozhuoxu, zhaoweijie12, laos1984, zhixin0825, pingli98\}@gmail.com}
}

\date{}

\maketitle
\begin{abstract}\vspace{0.3in}

\noindent\footnote{Currently, Zhaozhuo Xu is  a PhD student at Rice University, Weijie Zhao is  an Assistant Professor at Rochester Institute of Technology, and Zhixin Zhou is an Assistant Professor at the City University of Hong Kong. The work was conducted in 2018 and 2019 while Zhaozhuo Xu was a  research intern (and a graduate student at Stanford) and both Zhixin Zhou and Weijie Zhao were Postdoctoral Researchers, in Cognitive Computing Lab, Baidu Research.}Approximate Nearest Neighbor (ANN) search is a fundamental technique for (e.g.,) the deployment of recommender systems. Recent studies bring proximity graph-based methods into  practitioners' attention---proximity graph-based methods outperform other solutions such as quantization, hashing, and tree-based ANN algorithm families. In current recommendation systems, data point insertions, deletions, and queries are streamed into the system in an online fashion as users and items change dynamically.
As proximity graphs are constructed incrementally by inserting data points as new vertices into the graph, online insertions and queries are well-supported in proximity graph.
However, a data point deletion incurs removing a vertex from the proximity graph index, while no proper graph index updating mechanisms are discussed~in~previous~studies. To tackle the challenge, we propose an incremental proximity graph maintenance (IPGM) algorithm for online ANN. IPGM  supports both vertex deletion and insertion on proximity graphs. Given a vertex deletion request, we thoroughly investigate solutions to update the connections of the vertex. The proposed updating scheme eliminates the performance drop in online ANN methods on proximity graphs, making the algorithm suitable for practical systems.
\end{abstract}

\newpage\clearpage

\section{Introduction}
	Approximate Nearest Neighbor (ANN) search is a fundamental problem with wide applications in recommendation systems,  search engine, and advertising~\citep{pazzani2007content,vanderkam2013nearest,xu2018deep,fan2019mobius,tan2020fast,tan2021fast,yu2022egm}. Given a query vector $q$ and a dataset $\mathcal{D}$, the goal of ANN search is to find a vector $p\in \mathcal{D}$ that minimizes metric distance $f(p,q)$ with sub-linear time complexity. Here, $f$ represents a distance measure, e.g., Euclidean distance, cosine, Jaccard, GMM (generalized min-max)~\citep{li2018several} etc.  The research on ANN dated back to the  early days of computer science, i.e., KD-tree algorithms~\citep{friedman1975algorithm,friedman1977algorithm}.  Even recently, tree-based methods are still active research topics~\citep{lin1994tv,berchtold1996x,zezula1998approximate,sakurai2000tree,ram2019revisiting}.  Another popular line of research on ANN is based on hashing, including random sampling, (Gaussian and non-Gaussian) random projections, minwise hashing, consistent weighted sampling, extremal process, etc.~\citep{broder1997resemblance,broder1998minwise,indyk1998approximate,gionis1999similarity,charikar2002similarity,datar2004locality,li2005using,weiss2008spectral,shrivastava2012fast,wei2014scalable,erin2015deep,gao2015selective,shen2015supervised,zhu2016deep,li2017theory,li2019sign,li2019random,li2021consistent}. Quantization-based ANN methods are also commonly used~\citep{jegou2010product,ge2013optimized,zhang2014composite,andre2015cache,cao2016deep}.
	
	\vspace{0.1in}
	
	In this paper, we focus on graph-based ANN methods~\citep{paredes2005using,hajebi2011fast,malkov2014approximate,morozov2018non,fu2019fast,tan2019efficient,zhou2019mobius,malkov2020efficient,zhao2020song,tan2021norm}, which have successfully demonstrates their impressive empirical performance. Graph-based ANN search methods share the spirits with idea of six degrees of separation~\citep{karinthy1929chain,guare1990six,kleinberg2000small}. The general paradigm of these methods can be summarized as two steps: First, we index the dataset as a proximity graph $G$, where every vertex represents a data vector and every edge connects two ``close'' vertices. Then, we walk on $G$ towards the query $q$ to obtain its nearest neighbors: we maintain a priority queue that stores the searching vertex candidates. In each search step, we extract the vertex that is closest to $q$ from the priority queue and push all unvisited neighbors of the vertex into the priority queue for the future search steps.
	The proximity graph indexing is constructed in a bootstrap manner. We insert data vectors as new vertices incrementally. Each time a new vertex is inserted, we query on the existing graph index to obtain its nearest neighbors from the priority queue and create edges between them.

	However, graph-based methods do not support ANN search in online settings due to the lack of vertices update mechanism. In real-world applications, the online ANN search means that the dataset $D$ is updated by deletion and insertion for a period of time. For instance, in sponsored search systems for online advertising, the deep Click-Through Rate (CTR) prediction models~\citep{dave2010learning,fan2019mobius,yu2022egm} learn user vectors and Ad vectors and measure them by cosine distance to predict the final CTR. Given a user vector, ANN search is applied to retrieve ad vectors that minimize their cosine similarity. Due to the dynamic changes of sponsored items over time, the online Ads systems will update the ad vectors by deleting the expired Ads and inserting the recently arrived Ads.
    In these settings, although the incremental vertex insertion construction algorithm of graph-based algorithm can handle the upcoming Ads vectors, there is no specialized design for removing expired vertices on graph.  If we preserve the vertices of expired Ads on the graph, the search precision and efficiency will drop because a branch of time is wasted in visiting and ranking useless vertices. If we remove these vertices as well as their edges, the connectivity of the graph will be broken, which will also cause inefficiency in the search phase. There are few literatures on the graph-based online ANN search approaches. Therefore, when graph-based methods are deployed in online systems, a periodical reconstruction of indexing graph is inevitable. As re-indexing the large scale datasets scales as $\mathcal{O}(n\log(n))$~\citep{malkov2014approximate}, which increases the computation and memory cost for online services.

	To address those challenges, this paper studies online ANN search with graph-based indexing methods. We imitate the real-world settings by incrementally performing insertion and deletion in given ANN search vector sets. To dynamically handle the changes in data distribution, we propose an incremental proximity graph maintenance (IPGM) algorithm that supports online insertion and deletion of vertices on graph, which enables the indices update when the distribution of data is evolving periodically.
	
\vspace{0.1in}
	
\noindent\textbf{Summary of main contributions:} \ Firstly,  we formally define the proximity graph maintenance problem and theoretically prove the feasibility of incremental updates on proximity graphs. Then, we propose an online graph maintenance framework and 4 novel proximity graph update strategies that maintain proximity graph properties when deleting vertices. Massive experimental evaluations on public datasets confirm the effectiveness---our recommend global reconnection algorithm outperform the reconstruction baseline in query processing speed by a maximum factor~of~18.8$\%$.

	\section{Related Work}
	
\subsection{Graph-based ANN Search}	
Graph-based ANN search methods~\citep{paredes2005using,hajebi2011fast,malkov2014approximate,morozov2018non,fu2019fast,tan2019efficient,zhou2019mobius,malkov2020efficient,zhao2020song,tan2021norm} have received considerable research attention. Most ANN search scenarios studied by these papers are static. In the experiments section, the graph-based methods are compared with other ANN search methods on a fixed dataset such as Sift~\citep{jegou2010product} or Glove~\citep{pennington2014glove}.
In recommendations systems, one of the core tasks is: given a user embedding, to find relevant item embeddings. This is usually realized by learning an auxiliary data structure concurrently with learning item embeddings~\citep{zhu2018learning,zhu2019joint,zhuo2020learning,gao2020deep,yu2022egm}. \citep{Proc:Tan_WSDM20,tan2021fast} discover that graph-based ANN techniques can be leveraged on complex neural network similarity measure, i.e., user-item matching score, to efficiently retrieve recommendation candidates.
However, in current ANN applications like feed-in Ads recommendation~\citep{fan2019mobius}, the distribution of dataset (e.g., product vectors) changes over time due to the activation and expiration of products. Few works addressed this problem except~\citet{singh2021freshdiskann}. Therefore, modification to graph-based methods for online manner is essential for filling their gap with practical usage. However, as the insertion or deletion of vertex on graph breaks its topology,  direct modification of graph-based methods for online manner is challenging due to the unaffordable computational cost.

\subsection{Online ANN Search}

Besides graph-based ANN methods, some ANN methodologies can be modified for online settings, such as hashing methods and quantization methods. However, their shortcomings lie in the lower speed during the search phase.
	
\citet{ghashami2015binary,leng2015online,cakir2017online} modify the supervised hashing algorithms to accommodate the sequential inputs. These online hashing models achieve promising results in handling sequential input vectors. However, these online hashing requires expensive and infeasible labeling for each data vectors in real-world systems. Meanwhile, these online learning paradigms do not support the expired data vectors removal.

For quantization methods, the main challenge is to update the codebook according to the change of data distribution in database. 	In online settings, the changes in the distribution in database require expensive re-computation of codebook and codewords. To tackle this issue,~\citet{xu2018online} propose a PQ approach that supports data insertion and deletion by updating the codebook without changing the codewords of existing vectors in the database. This online PQ method achieves promising results in sliding window settings.

    \section{Problem Statement}\label{sec:problem}
    Formally, the ANN search problem is defined as follows:
    \begin{problem}[\bf ANN Search]\label{prob:ann}
    given a dataset $D$ and a query vector $q$ in Euclidean space $\mathbb R^d$ we aim at efficiently computing
    \begin{align}\label{eq:ann}
        p=\arg\max_{x\in D} f(x,q).
    \end{align}

    \end{problem}

    Here, we discuss metric similarity such as $f(x,q)=-\|x-q\|$, which minimizes the Euclidean distance, or $f(x,q)=\frac{x^Tq}{\|x\|\|q\|}$, which maximizes the cosine similarity. In most real-world datasets, linear scan is infeasible to solve this problem because $|D|\gg d$. To remedy this issue, we index $D$ as a proximity graph $G=(D,E)$, where each vertex represents a $x_i \in D$ and each pair $(x_i,x_j)$ indicates the connection between vertices. Given each $q$, we find the solution of Eq.~\eqref{eq:ann} by walking on $G$. As a walking path only covers a fraction of vertices on $G$, the speed for solving Eq.~\eqref{eq:ann} is accelerated at the expense of precision.

    For metric similarities that satisfy triangle inequality, an efficient proximity graph $G$ usually approximates \emph{Delaunay graph}, which is defined by \emph{Voronoi cells}.

\vspace{0.1in}

    \begin{definition}\label{def:voronoi}
    For fixed $x_i\in D\subset \mathbb R^d$ and a given function $f$, the Voronoi cells $R_i$ is defined as
    \begin{align*}
    R_i := R_i(f,D):= \{q\in \mathbb R^d\ |\ \forall x\in D, f(x_i, q)\ge f(x,q)\}.
    \end{align*}
    \end{definition}

\vspace{0.1in}

    \emph{Voronoi cells} determine the solution of Eq.~\eqref{eq:ann}. We observe that, $x_j\in \arg\max_{x_i\in D} f(x_i, q)$ if and only if $q\in R_j$. The \emph{Delaunay graph} is the dual diagram of \emph{Voronoi cells}, which is given by the following definition.

\vspace{0.1in}

    \begin{definition}\label{def:delaunay:graph}
        For fixed function $f$ and dataset $D\subset \mathbb R^d$, and given \emph{Voronoi cells} $R_i$, $i = 1, 2, \dots, n$ w.r.t. $f$ and $D$,
        the Delaunay graph of $D$ is an undirected graph $G$ with vertices $\{x_i \in D\}$, and the edge $\{x_i, x_j\}$ exists if $R_i\cap R_j\ne\emptyset$.
    \end{definition}
\vspace{0.1in}

    For graph-based ANN search problem, the guarantee of finding the optimal solution for Eq.~\eqref{eq:ann} on \emph{Delaunay graph} is given in Theorem~\ref{thm:main}, with proof provided in~\citet{morozov2018non}.

\newpage

\begin{theorem}\label{thm:main}
    For given metric similarity $f$, we assume for any dataset $D$, each the Voronoi cell $R_i$ is a connected. Let $G = (D,E)$ be the \emph{Delaunay graph} w.r.t. the Voronoi cells. Then for any $q\in \mathbb R^d$, greedy search on \emph{Delaunay graph} returns the solution of \eqref{eq:ann}. In other words, let $N(x_i) = \{x_j\in D: \{x_i, x_j\}\in E\}$ be the neighbors of $x_i$ on $G$. If $x_i$ satisfies
    \begin{align}\label{eq:main:result}
    f(x_i, q) \ge \max_{x_j\in N(x_i)} f(x_j, q),
    \end{align}
    then $x_i$ is a solution of Eq.~\eqref{eq:ann}.
    \end{theorem}

    For large scale dataset, the computation cost of constructing and storing both \emph{Voronoi cells} and \emph{Delaunay graph} are unaffordable. The performance of proximity graphs in ANN search depends on their approximation to \emph{Delaunay graph}.

    In this work, we discuss the maintenance of $G$ in online settings, where the ANN search problem becomes:

     \begin{problem}[\bf Online ANN Search]\label{prob:mips}
    Given a dataset sequence $\{D_1,D_2, \dots D_T | D_i \cap D_{i+1} \neq \emptyset, i=1,2,\dots, T\}$, and a query vector $q$ in Euclidean space $\mathbb R^d$ we aim at efficiently computing
    \begin{align}\label{eq:ann_online}
        p=\arg\max_{x\in D_i} f(x,q) \quad i\in \{1,2,\dots, T\}
    \end{align}

    \end{problem}

    Here, $D_i \cap D_{i+1} \neq \emptyset$ indicates the possibility of data insertion and deletion in each step. Start from proximity graph $G_1$ that approximates \emph{Delaunay graph} of $D_1$, our goal is to develop a proximity graph maintenance algorithm $\emph{update}(G,D)$ so that $G_{i+1}=\emph{update}(G_i, D_{i+1}) \quad i\in \{1,2,\dots, T\}$ is an approximation of \emph{Delaunay graph} with respect to $D_{i+1}$.

    \section{Incremental Proximity Graph}
    \label{sec::ann}

We first theoretically prove that deletion a vertex on Delaunay Graph only requires update of the removed vertex's neighbors' connections. Based on this insight, we propose a high-level incremental proximity graph maintenance framework for online update~settings.

    \subsection{Theoretical Analysis}
    In this section, we prove that we only require to update the edges of the deleted vertex's neighbors to maintain the properties of Delaunay graph---the incremental graph maintenance is feasible in online settings.
    Given the Delaunay graph $G$ of vertices $S$, suppose we remove an vertex $x_0$ from $S$, and
    \begin{enumerate}

    \item suppose $x\in D$ is not a neighbor of $x_0$ in $G$, then the neighbors of $x$ in the new Delaunay graph (corresponding to $S\backslash\{x_0\}$) remains the same as the ones in the old $G$.
    \item suppose $x\in D$ is a neighbor of $x_0$ in $G$, then it's incident edges except $\{x, x_0\}$ remains existing in the Delaunay graph.
    \end{enumerate}
    These properties are rigorously stated in the following theorem.
    Given a graph $G$, $G$ is a \emph{Delaunay graph} or contains \emph{Delaunay graph} as subgraph, we demonstrate that
    \begin{enumerate}
        \item If a vertex $x_i$ is removed from $G$, $G$ still contains a  \emph{Delaunay graph} if we update the connections of $N(x_i)$.
        \item If a vertex $x_i$ is inserted to $G$, $G$ still contains a  \emph{Delaunay graph} if we set up $N(x_i)$.
    \end{enumerate}

    \begin{theorem}\label{thm:update}
        Let $f$ be continuous function such that for any datasets, every Voronoi cell is nonempty and connected. Let $G$ and $G'$ be the Delaunay graphs corresponding to data points $D = \{x_0, \dots, x_n\}$ and $D'=\{x_1, \dots, x_n\} = D\backslash\{x_0\}$. Let $N(x_i, G)$ and $N(x_i, G')$ be the set of neighbors of $x_i$ on graph $G$ and $G'$ respectively, then
        \begin{itemize}[leftmargin=0.23in,nosep]
        \item[(a)]for all $x_i\notin N(x_0, G)$, $N(x_i, G)=N(x_i, G')$;
        \item[(b)]removing $x_0$ and its the incident edges from $G$, the resulting graph is a subgraph of $G'$.
        \end{itemize}
    \end{theorem}

\begin{proof}
(a) We consider two cases.

Firstly, if $N(x_i, G)\subset N(x_i, G')$, we show that, by the definition of Voronoi cells, for all $x\in D'$ where $D'\subset D$, we have
\begin{align*}
    R_x(D)  =&  \bigcap_{y\in D}\{ q\in Y: f(x,q)\ge f(y,q) \} =R_x(D')\\
    &\cap \{ q\in Y: f(x,q)\ge f(x_0,q) \}.
\end{align*}

Hence $R_x(D)\subset R_x(D')$. As a result, $R_x(D)\cap R_y(D)\ne\emptyset$ implies $R_x(D')\cap R_y(D')\ne\emptyset$. Thus $N(x_i, G)\subset N(x_i, G')$.\\

Secondly, if $N(x_i, G')\subset N(x_i, G)$, suppose $y\in N(x_i, G')$ but $y\notin N(x_i, G)$, then by definition of Delaunay graph, we have
\begin{align*}
    R_y(D')\cap R_{x_i}(D')\ne \emptyset
    \quad \text{ and } \quad
    R_y(D)\cap R_{x_i}(D)=\emptyset.
\end{align*}

Since
\begin{align*}
R_y(D)\cap R_{x_i}(D) =& R_y(D')\cap R_{x_i}(D')  \cap\{ q\in Y: f(y,q)\ge f(x_0,q) \} \\
&\cap\{ q\in Y: f(x_i,q)\ge f(x_0,q) \}
=\emptyset
\end{align*}

we have
\begin{align*}
    R_y(D')\cap R_{x_i}(D')\subset &
    \{\{ q\in Y: f(y,q)< f(x_0,q) \}\\
    &\cup\{ q\in Y: f(x_i,q)< f(x_0,q) \}\}
\end{align*}

Hence $R_y(D')\cap R_{x_i}(D')\subset R_{x_0}(D)$. Pick $q\in R_{x_i}(D)$ and $q'\in R_y(D')\cap R_{x_i}(D')$. Since $R_{x_i}(D')$ is connected, there exists a path $c:[0,1]\to R_{x_i}(D')$ connecting $q$ and $q'$. Considering the function
\begin{align*}
  g(t) = f(x_0, c(t)) - f(x_i, c(t))
\end{align*}

we have
\begin{align*}
  g(0)=f(x_0,q)-f(x_i,q)\le 0\qquad g(1)=f(x_0,q')-f(x_i,q')\ge 0
\end{align*}

  By continuity of $f$ and intermediate value theorem, there exists $t^*\in[0,1]$ such that $g(t^*)=0$.    Hence $R_{x_i}(D)$ and $R_{x_0}(D)$ intersect at $c(t^*)$, which contradicts with $x_i\notin N(x_0,G)$.
\\

(b) Suppose $x,y\in D$ and the edge $\{x,y\}\in G$, then $R_x(D)\cap R_y(D)\ne\emptyset$. Hence, $R_x(D)\cap R_y(D)\subset R_x(D')\cap R_y(D')\ne\emptyset$. Therefore, $\{x,y\}\in G'$.
\end{proof}

\newpage

	It was demonstrated in~\citet{navarro2002searching} that greed search on proximity graph is guaranteed to find the exact nearest neighbor for every given query. However, it is computationally infeasible to build exact Delaunay Graph in high-dimensional data because the number of edges grows exponentially as the scale of dataset increases~\citep{beaumont2007peer}. A typical solution towards this issue is to approximate Delaunay Graph and trade accuracy for significant efficiency improvements. Practical applications employ the proximity graph as an approximation of the Delaunay graph. Therefore, we focus on the incremental maintenance on proximity graphs in this paper.

\subsection{High-Level Proximity Graph Maintenance}

	In this section, we introduce our high-level incremental proximity graph maintenance algorithm (Algorithm~\ref{alg:high-level}).

	In each state of the workload, we approximate the Delaunay Graph via a directed proximity graph, where each vertex have edges pointed to its neighbors, namely out-neighbors. We use $W=\{(o_1,x_1),(o_2,x_2),\dots, (o_n,x_n)\}$ to represent the online ANN search workload. For each $(o_i,x_i)\in W$, $o_i$ is the type of the operations--namely query, insert, or delete--and $x_i$ is the corresponding vectors.

\vspace{0.1in}
	\textbf{Query.}
    Querying operation is performed by executing the greedy search on the proximity graph $G$. The algorithm for greedy search is present in  Algorithm~\ref{alg:greedys}. Given a query vector $q$ and a graph $G$, the algorithm first chooses a vertex $x_s$ randomly on $G$. Starting from $x_s$, $q$ walks on $G$. During the walking, we preserve a priority queue where the vertex that maximizes the measurement function $f$ is on top of the queue. For every visited vertex $x$, we push its out-neighbors in the queue after calculating their similarity with $q$.\vspace{0.2in}

            	\begin{algorithm}[h]
                \caption{GREEDY-SEARCH$(G,q,k,Y)$}\label{alg:greedys}
                \begin{algorithmic}[1]
                \STATE\textbf{Input:} graph $G = (D,E)$, query vector $q$, candidate size $k$, removed vertices $Y$.
                \STATE Initialize the priority queue by random sampling, $S\gets \{x_i| x_i\in G, x_i\notin Y\}$.
                \STATE Mark elements in $\{x_i| x_i\in G, x_i\notin Y\}$ as unvisited.
                \STATE Mark $x_i$ as visited.
                \IF {$|S|>k$}
                \STATE $S\gets$ top-$k$ vectors in $S$ that has the largest value of $-\|x-q\|$. We place them in a descending order.
                \ENDIF
                \WHILE {$\exists x\in D$ unvisited and $S$}
                \STATE $S\gets S \cup \{y\in D: x\in S, y\text{ unvisited}, (x,y)\in E, y\notin Y \}$
                \STATE Mark $S$ as visited.
                \IF {$|S|>k$}
                \STATE $S\gets$ top-$k$ vectors in $S$ that has the largest value of $-\|x-q\|$. We place them in a descending order.
                \ENDIF
                \ENDWHILE
                \STATE \textbf{Output:} $C$.
                \end{algorithmic}
                \end{algorithm}

\newpage

    \begin{algorithm}[t]
    \caption{SELECT-NEIGHBORS$(x,S,d,I)$}\label{alg:select}
    \begin{algorithmic}[1]\label{egde_select}
    \STATE\textbf{Input:} vector $x$, $k$-neighbor set $S$ of $x$, out-neighbor degree threshold $d$, invalid set $I$.
    \STATE Set $N_x$ as $\emptyset$.
    \STATE Place $y_i\in S$ in descending order of $-\|x-y_i\|$.
    \STATE $i\gets 1$.
    \WHILE {$|N_x|\le d$ and $i\le |S|$}
    \IF {$\|x-y_i\|\le \min_{z\in N_x}\|z-y_i\|$ and $y_i \notin I$}
    \STATE $N_x\gets N_x\cup\{y_i\}$.
    \ENDIF
    \STATE $i\gets i+1$.
    \ENDWHILE
    \STATE\textbf{Output:} a set of elements $N_x$.
    \end{algorithmic}
    \end{algorithm}

		\begin{algorithm}[h]
    \caption{GRAPH-MAINTENANCE$(W,k,d,f)$}
    \label{alg:high-level}
    \begin{algorithmic}[1]
    \STATE\textbf{Input:} workload $W=\{(o_1,x_1),\cdots, (o_n,x_n)\}$, priority queue length $k$, out-neighbor degree threshold $d$,measurement function $f$
    \STATE $Y\gets\{\}$.
    \FOR{$i = 1$ to $n$}
    \IF {$o_i=query$}
    \STATE $C_i \gets \text{\small GREEDY-SEARCH}(x_i, G,k, f,Y)$.
    \ELSIF {$o_i=Insert$}
    \STATE $C \gets \text{\small GREEDY-SEARCH}(x_i, G,k, f,Y)$.
    \STATE $N \gets \text{\small SELECT-NEIGHBORS}(x_i, C, d,\emptyset)$.
    \FOR{$z\in N$}
    \STATE Add edges $(x_i,z)$ to $G$ and $(z,x_{i})$ to $G'$.
    \ENDFOR
    \ELSIF {$o_i=delete$}
    \IF{$x_i\in G$}
    \STATE \text{\small DELETE-UPDATE-EDGES}
    \ENDIF

    \ENDIF
    \ENDFOR
    \STATE\textbf{Output:} top-$K$ objects $C_i\subset D$ in descending order of $f$ with $x_i$ for every $\{(o_i,x_i)\in W| o_i=query\}$.
    \end{algorithmic}
    \end{algorithm}

\textbf{Insertion.}
When $o_i$ is inserted, we perform a greedy search on $G$ to obtain $x_i$'s Top-K neighbors (same as in the query phase). We then extract the vertex from the queue as the start point for next search step. When the walking stops, we select the Top-K vertices in the queue to return. The next step is to select the Top-K neighbors of $x_i$ found on $G$ and set connections in $G$ and $G'$. Here we adapt the edge selection algorithm~\citep{malkov2014approximate} (Algorithm~\ref{egde_select}). After the selection, $x_i$'s outgoing vertices are its nearest neighbors in diverse~directions.

\newpage

    \textbf{Deletion.}
    If $o_i$ is delete and $x_i\in G$, we start the vertex deletion and graph update process. For every expired vertex $x_i$, we first obtain its incoming vertices on $G$, by the prepared reversed graph $G'$. Let $N'(x_i)$ be the in-neighbors of  vertex $x_i$ (the vertices that points to $x_{i}$ in $G$). \text{DELETE-UPDATE-EDGES} updates the edges of $N'(x_{i})$ to maintain  the proximity graph---the nearest neighbors can be retrieved from the updated graph by the greedy searching algorithm.
    Besides, the reverse graph $G'$ is updated according to $G$.

\vspace{0.1in}

    A sub-optimal updating algorithm reduces the connectivity of the proximity graph and hence the greedy searching cannot locate accurate nearest neighbors.
    In the online settings, the deletion operation is called regularly in the workload---the disadvantage of sub-optimal updating algorithms are magnified.
    Therefore, we delve into the design and implementation of the \text{DELETE-UPDATE-EDGES} algorithm in the following paper.

\section{DELETE-UPDATE-EDGES Algorithms}
    We investigate online proximity graph deletion algorithms in this section. We begin with straightforward solutions and then keep refining the update algorithm step by step.

\begin{figure}[h]
\centering
\includegraphics[width=6in]{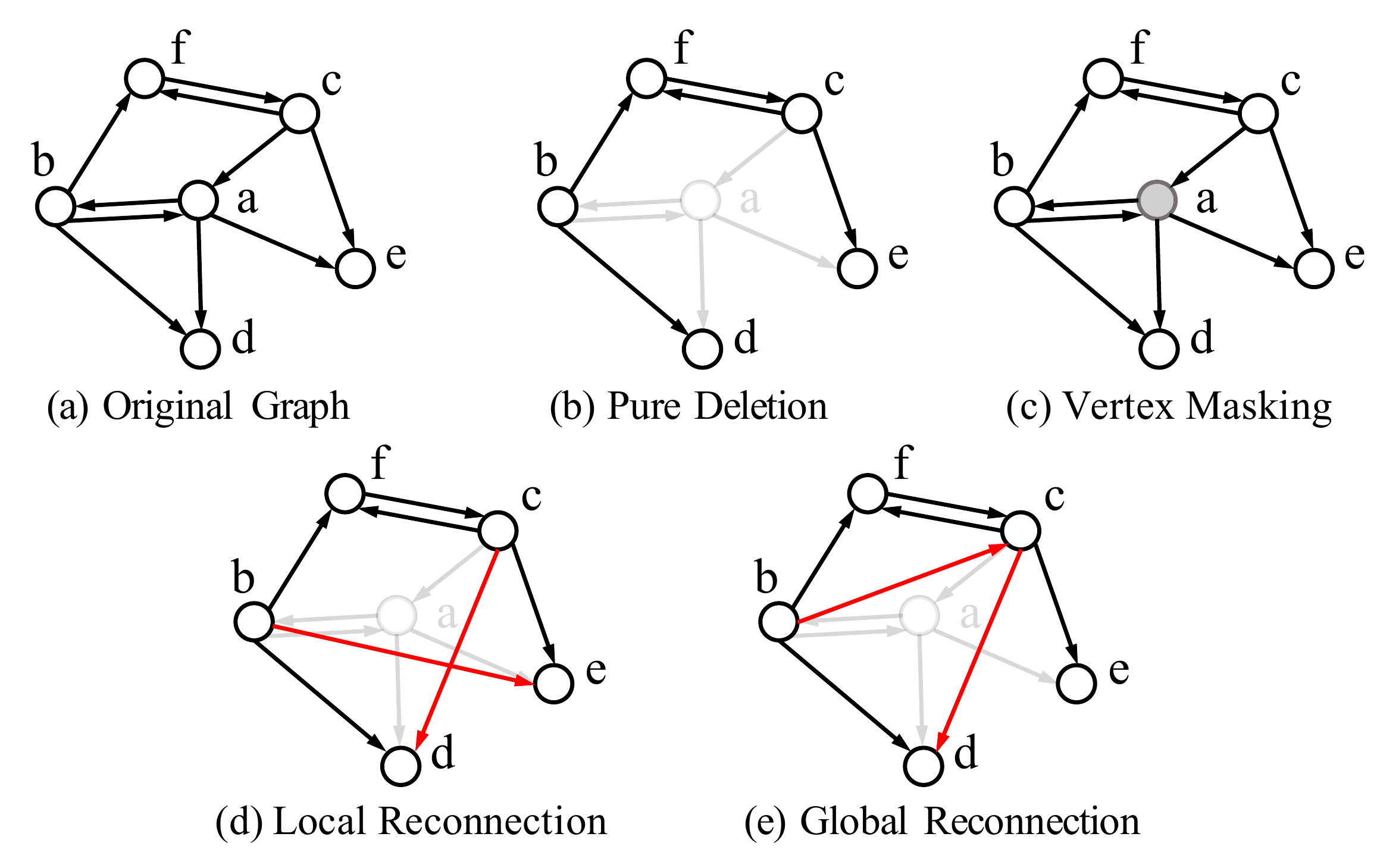}

\caption{Illustrations of graph updating algorithms.}
\label{fig:diagram}
\end{figure}

    \subsection{Pure Delete}
    The most straightforward delete operations for proximity graph update is to remove the vertex as well as its connections. Here we demonstrate our pure deletion algorithm in Algorithm~\ref{alg:pure}  and Figure~\ref{fig:diagram} (b). Given a deleted vertex $x_i$, we first obtain its out-neighbors and in-neighbors. For each out-neighbor $x_j$, we remove edge $(x_i,x_j)$ on proximity graph $G$ and $(x_j,x_i)$ on the reverse proximity graph $G'$. For each in-neighbor $x_k$, we remove edge $(x_k,x_i)$ on proximity graph $G$ and $(x_i,x_k)$ on the reverse proximity graph $G'$. Pure deletion breaks the connectivity of the proximity graph, which may lead to the fail of greedy search as the path towards the exact nearest neighbors are blocked.

   \subsection{Vertex Masking}
    To tackle the preserve the connectivity challenge in the \texttt{Pure Deletion} method, one straightforward solution is to mask the deleted vertices as ``deleted'' and push it into $Y$ in Algorithm~\ref{alg:greedys}. When we perform a greedy searching, the ``deleted'' vertices are still visited. However, we do not push the ``deleted'' vertices into the result priority queue---the visited ``deleted'' vertices are not counted as top-$K$ candidates. The connections of the vertices in the graph are unchanged. We illustrate this method in Figure~\ref{fig:diagram} (c).

    One advantage of the masking method is that the connectivity of the original proximity graph is preserved and the expired vertices are still useful for guiding the greedy search. However, the disadvantages are still obvious. First, after masking, the search space for each query is large and unnecessary. The visit to expired vertices slows down the greedy search speed. Second, as all expired vertices are still stored, space grows continuously as the stream performs, which may cause inevitable memory issues.

	\begin{algorithm}[t]
    \caption{PURE-DELETE$(x_i, G, G', f, d)$} \label{alg:pure}
    \begin{algorithmic}[1]
    \STATE\textbf{Input:} deleted vertex $x_i$, proximity graph $G$ and its reverse graph $G'$, measurement function $f$, maximum outgoing degree of graph $d$.
    \STATE $N'(x_i)\gets$ out-neighbors of $x_i$ in $G'$.
    \STATE $N(x_i)\gets$ out-neighbors of $x_i$ in $G$.
    \FOR{$x_j$ in $N(x_i)$}
    \STATE Remove $(x_i, x_j)$ in $G$ and $(x_j, x_i)$ in $G'$
    \ENDFOR
    \FOR{$x_k$ in $N'(x_i)$}
    \STATE Remove $(x_k, x_i)$ in $G$ and $(x_i, x_k)$ in $G'$
    \ENDFOR
    \STATE\textbf{Output:} $G$, $G'$.
    \end{algorithmic}
    \end{algorithm}

  \begin{algorithm}[t]
    \caption{LOCAL-RECONNECT$(x_i, G, G', f, d)$} \label{alg:reconnect_local}
    \begin{algorithmic}[1]
    \STATE\textbf{Input:} deleted vertex $x_i$, proximity graph $G$ and its reverse graph $G'$, measurement function $f$, maximum outgoing degree of graph $d$.
    \STATE $N'(x_i)\gets$ out-neighbors of $x_i$ in $G'$.
    \STATE $N(x_i)\gets$ out-neighbors of $x_i$ in $G$.
    \FOR{$x_j$ in $N'(x_i)$}
    \STATE $N(x_j)\gets$ out-neighbors of $x_j$ in $G$.
    \STATE $z \gets \text{\small SELECT-NEIGHBORS}(x_j, N(x_i), 1, N(x_j)\cup \{x_j\})$.
    \STATE Remove $(x_j,x_i)$ in $G$ and $(x_i,x_j)$ in $G'$
    \IF{$z!=\mathsf{null}$}
    \STATE Add edges $(x_j,z)$ to $G$ and $(z,x_{j})$ to $G'$
    \ENDIF
    \ENDFOR
    \STATE\textbf{Output:} $G$, $G'$.
    \end{algorithmic}
    \end{algorithm}

    \subsection{Local Reconnect}

    To remedy the connectivity of the proximity graph while reducing the useless visit in the search phase, another solution is to reset local connections between the in-neighbors and out-neighbors of deleted vertex. The local reconnect algorithm is presented in Algorithm~\ref{alg:reconnect_local}  and Figure~\ref{fig:diagram} (d). For each expired vertex $x_i$, we first acquire its in-neighbors $N'(x_i)$ from $G'$ and out-neighbors $N(x_i)$ from $G$. Then, for each in-neighbor $x_j$, we remove the edge from $x_j$ to $x_i$ in $G$ and $G'$. The last step is to select the most diverse vertex $z$ for $x_j$ from $N(x_i)$ and set out-going links from $x_j$ to $z$.

    The idea of local reconnect is to compensate an edge for the in-neighbors of deleted vertex. According to the nearest neighbor expansion, the out-neighbors of expired vertex $x_i$ may be the neighbors of $x_j\in N'(x_i)$. However, as the proximity graph is an approximation to real Delaunay Graph, out-neighbors of $x_i$ may not be $x_i$'s exact nearest neighbors. Therefore, the compensated edge may still not be good.

    \subsection{Global Reconnect}
    To compensate the in-neighbors of deleted vertices with better edges, we propose a global reconnection algorithm for proximity graph update. The algorithm is presented in Algorithm~\ref{alg:global} and Figure~\ref{fig:diagram} (e). For each expired vertex $x_i$, we first acquire its in-neighbors $N'(x_i)$ from $G'$ and out-neighbors $N(x_i)$ from $G$.  After that, for every element $x_j$ in $N'(x_i)$, we perform greedy search to obtain its nearest neighbors $C$ on graph $G$. Then we use \texttt{SELECT-NEIGHBORS} function to select the vertices set $N$ from $Z$ to fit the regulation for graph degree $d$. Finally, we remove all the previous connections of $x_j$ and reset edges by vertices in $N$.

    The general idea behind this deletion method is to treat each in-neighbor $x_j$ of the expired vertex as a new vertex and reset $x_j$'s connections by re-inserting it on proximity graph. In this way, $x_j$ has an opportunity to improve its connection by re-selecting vertices globally, Therefore, the proximity graph may become a better approximation to Delaunay Graph. The disadvantage is the updating time cost. however, it can be amortized by the number of query operations. In real-world recommendation systems or search engine, the number of queries(users) may be much larger than the number of data(Ads, items) and one query may occur in high frequency. Thus, this deletion method may outperform peers in online graph-based ANN search.

	\begin{algorithm}[t]
    \caption{GLOBAL-RECONNECT$(x_i, G, G', k, f,  d)$}\label{alg:global}
    \begin{algorithmic}[1]
    \STATE\textbf{Input:} deleted vertex $x_i$, proximity graph $G$ and its reverse graph $G'$, priority queue length $k$, measurement function $f$, out-neighbor degree threshold $d$.
    \STATE $N'(x_i)\gets$ out-neighbors of $x_i$ in $G'$.

    \FOR{$x_j$ in $N'(x_i)$}
    \STATE $C \gets \text{\small GREEDY-SEARCH}(x_j, G,k, f)$.
    \STATE $N \gets \text{\small SELECT-NEIGHBORS}(x_j, C, d,\{x_i\})$.
    \STATE Remove $N(x_j)$ in $G$ and $N'(x_j)$ in $G'$
    \FOR{$z\in N$}
    \STATE Add edges $(x_j,z)$ to $G$ and $(z,x_j)$ to $G'$.
    \ENDFOR
    \ENDFOR
    \STATE\textbf{Output:} $G$, $G'$.
    \end{algorithmic}
    \end{algorithm}

    \section{Experiments}\label{sec:experiments}
    The objective of the experimental evaluation is to investigate the performance of our update algorithms on real datasets.
    The performance includes both the quality and efficiency of the retrieving top-$K$ candidates and the execution time for indexing and searching.
    Specifically, we target to answer the following questions:

    \begin{itemize}[leftmargin=*]
      \setlength\itemsep{0.05em}
    \item Does the proposed global reconnect algorithm improve the performance of the proximity graph? How does it compare with other graph update algorithms?

    \item Does the proposed global reconnect algorithm robust to different update patterns? How is the proposed algorithm and baselines perform in different update patterns?

    \item Does the proposed global reconnect algorithm reduce the total execution time for insertion, deletion, and query?

    \end{itemize}

\textbf{Implementation.}
    We implement the update algorithms as a \texttt{C++11} prototype and compile the code with g++-5.4.0 and ``O3'' optimization. The implementation in~\citet{zhao2020song} is employed as the proximity graph searching algorithm. The code contains special functions to harness detailed profiling data.

    \textbf{Hardware System.}
    We execute the experiments on a single node server. The server has one Intel(R) Core(TM) i7-5960X CPU @ 3.00GHz (64 bit). It has 8 cores 16 threads, and 32 GB memory. Ubuntu 16.04.4 $64$-bit is the operating system.

	\textbf{Retrieval Recall}.
	In this paper, we use recall to measure the quality of search, which is defined as the ratio of retrieved correct items over the total correct items. The increase in recall relates to higher quality.

    \textbf{Data.}
    We use $4$ real ANN benchmark datasets for experiments: SIFT~\citep{jegou2011searching}, GloVe200~\citep{pennington2014glove}, NYTimes~\citep{Dua:2019} and GIST~\citep{sandhawalia2010searching}. The dimension of vectors in these datasets are 128,200,256 and 960. The NYTimes contains $280,000$ vectors while others contain $1,000,000$ vectors. We have a diverge distribution of the datasets---GloVe200 and NYTimes are more skewed compare to others.

	\textbf{Workload.}
    For each of the 4 ANN datasets, we build 10 step workloads. Given the base set, each step of the workload removes a set with 10,000 vectors and then digests an set with 10,000 vectors. After that, a query set of 10,000 vectors are fed for top-K ANN search. In this case, the size for base set in NYTimes is 180,000 while other sets are 900,000.

    We consider two update patterns in the experiments:
    (a) Random updates: we permute the whole data set before partition. Each vector of base, delete, insert and query set is randomly selected from the original dataset.
    (b) Cluster updates: we perform 10 class K-means clustering on the whole dataset. Then we place the clusters in a sequence and build base, delete and insert sets in the order of this sequence. Therefore, each step we remove several clusters of vectors on graph. And then insert several clusters.

    \textbf{Methods.} The 4 proposed algorithms are abbreviated as PURE, MASK, LOCAL and GLOBAL. We also include ReBuild, which reconstructs the whole graph in each update batch~before~query.

\newpage

\begin{figure}[h!]

\mbox{\hspace{-0.1in}
\includegraphics[width=2.3in]{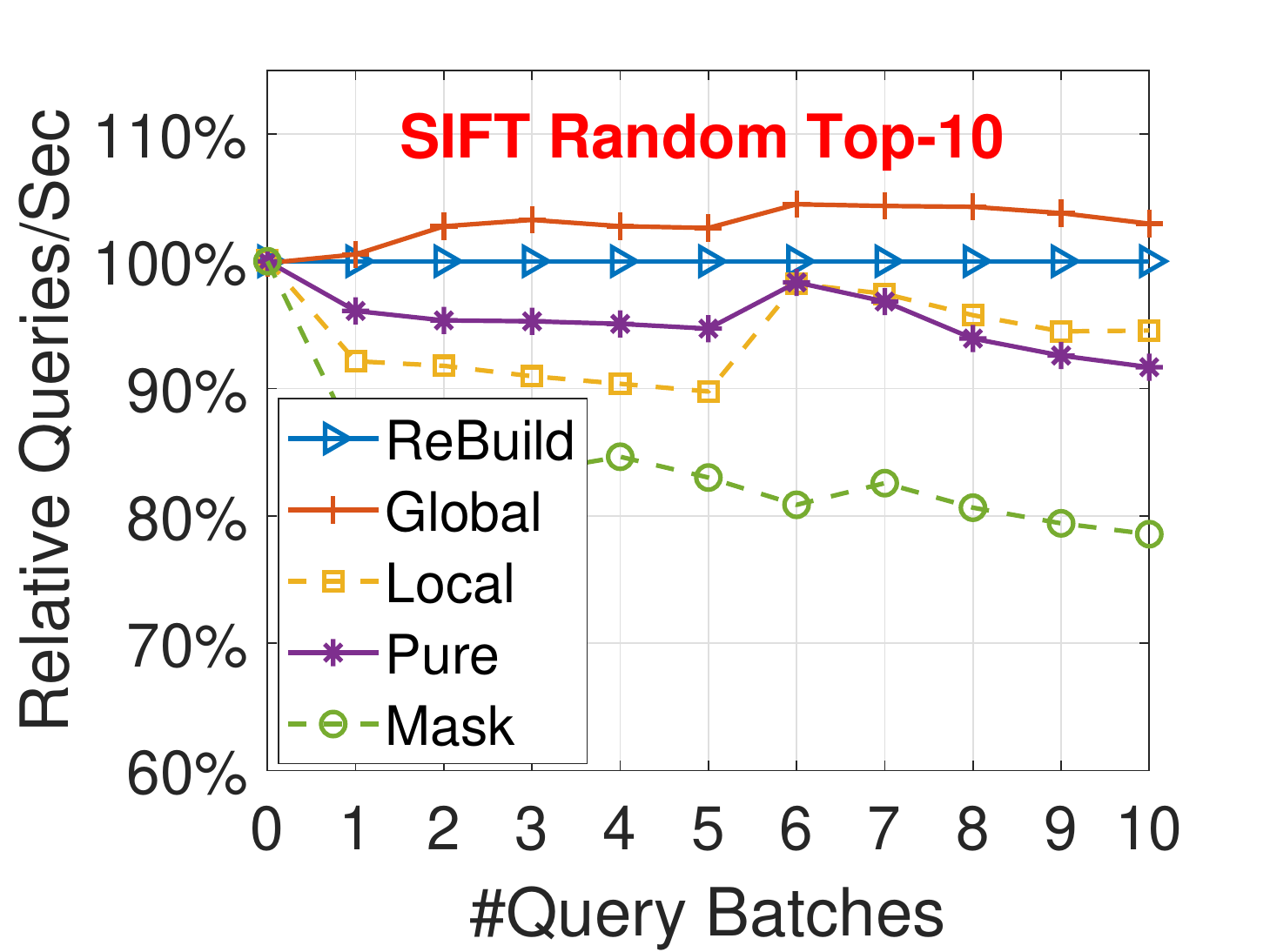}\hspace{-0.1in}
\includegraphics[width=2.3in]{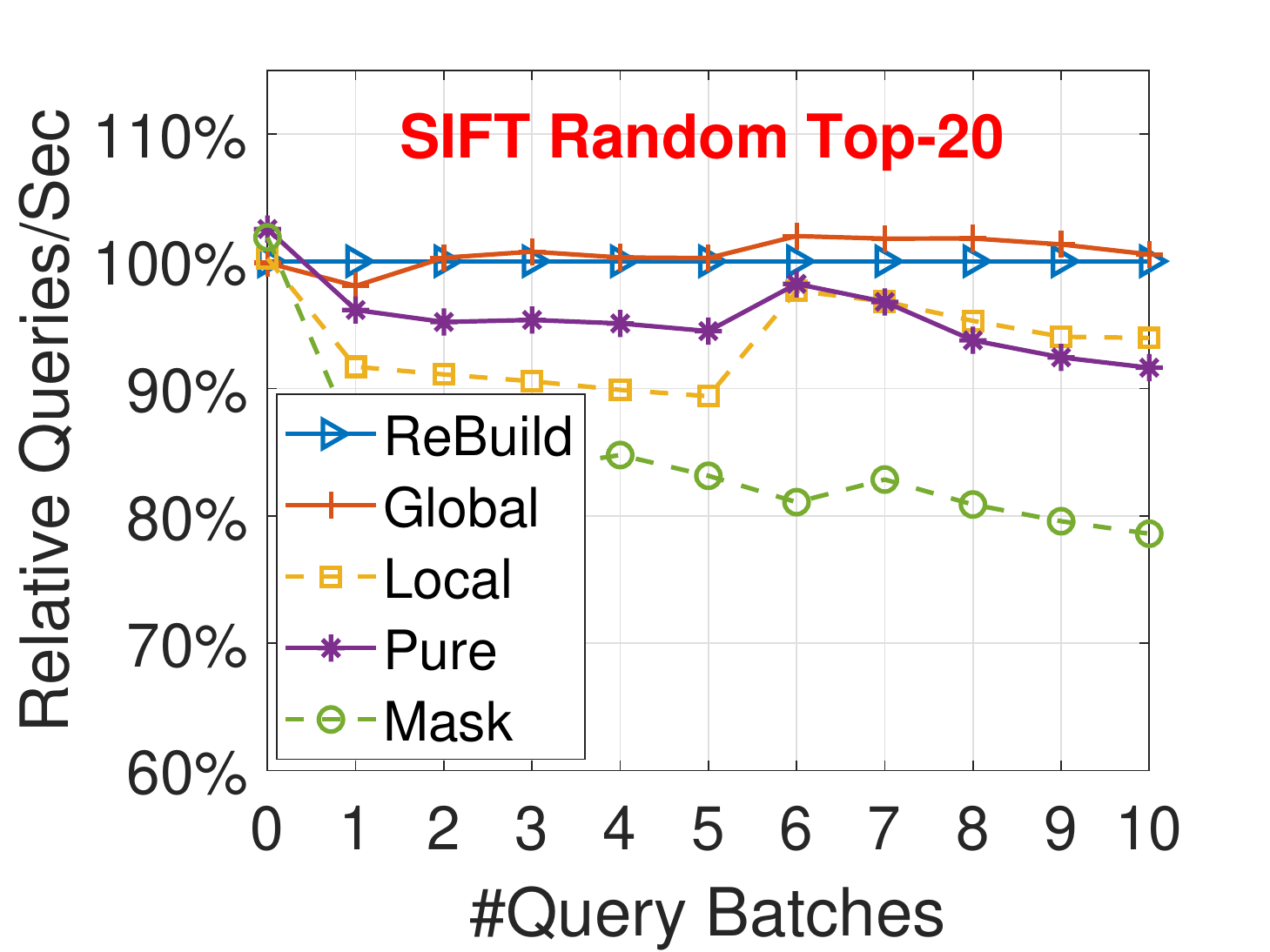}\hspace{-0.1in}
\includegraphics[width=2.3in]{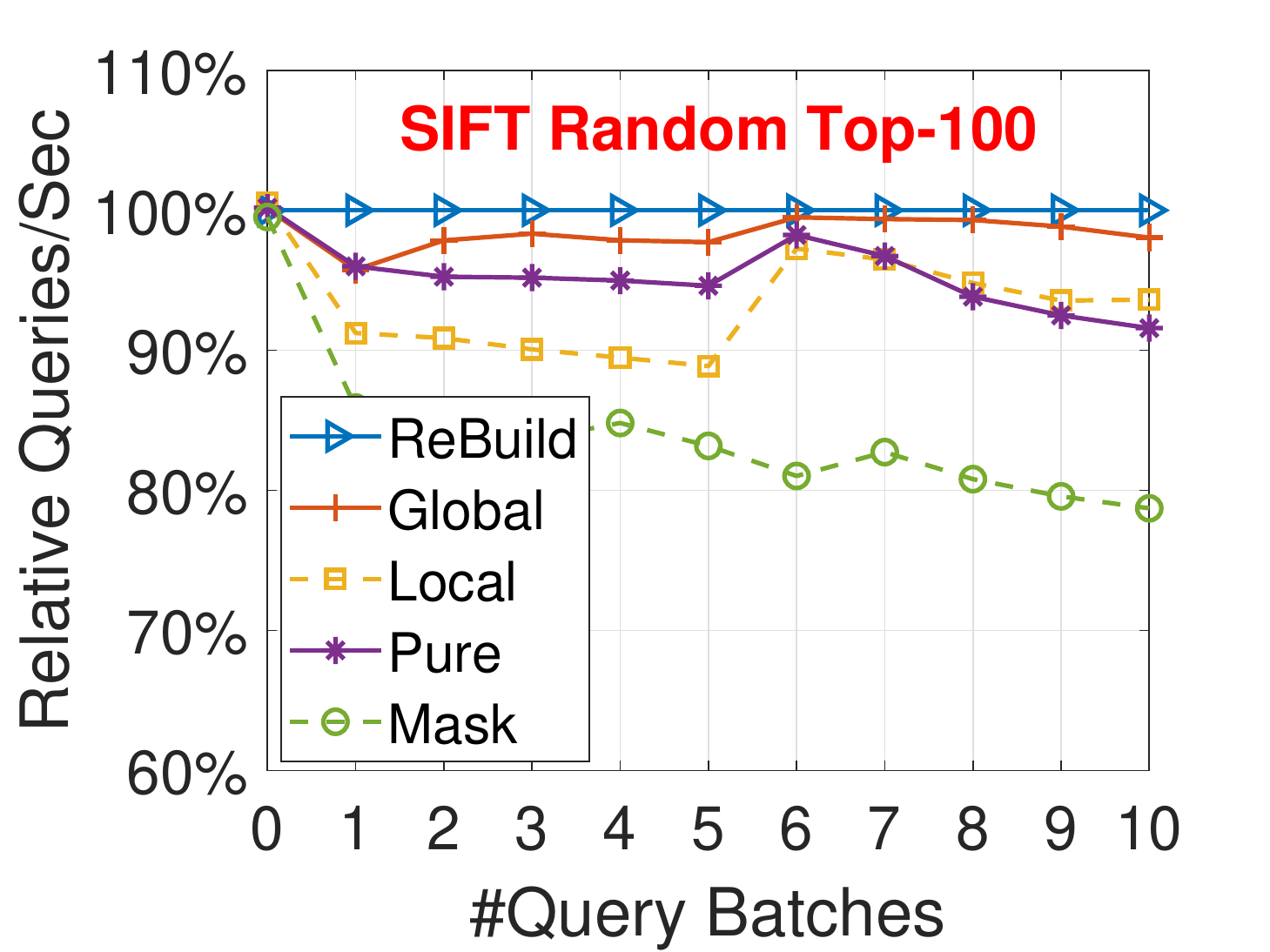}
}

\mbox{\hspace{-0.1in}
    \includegraphics[width=2.3in]{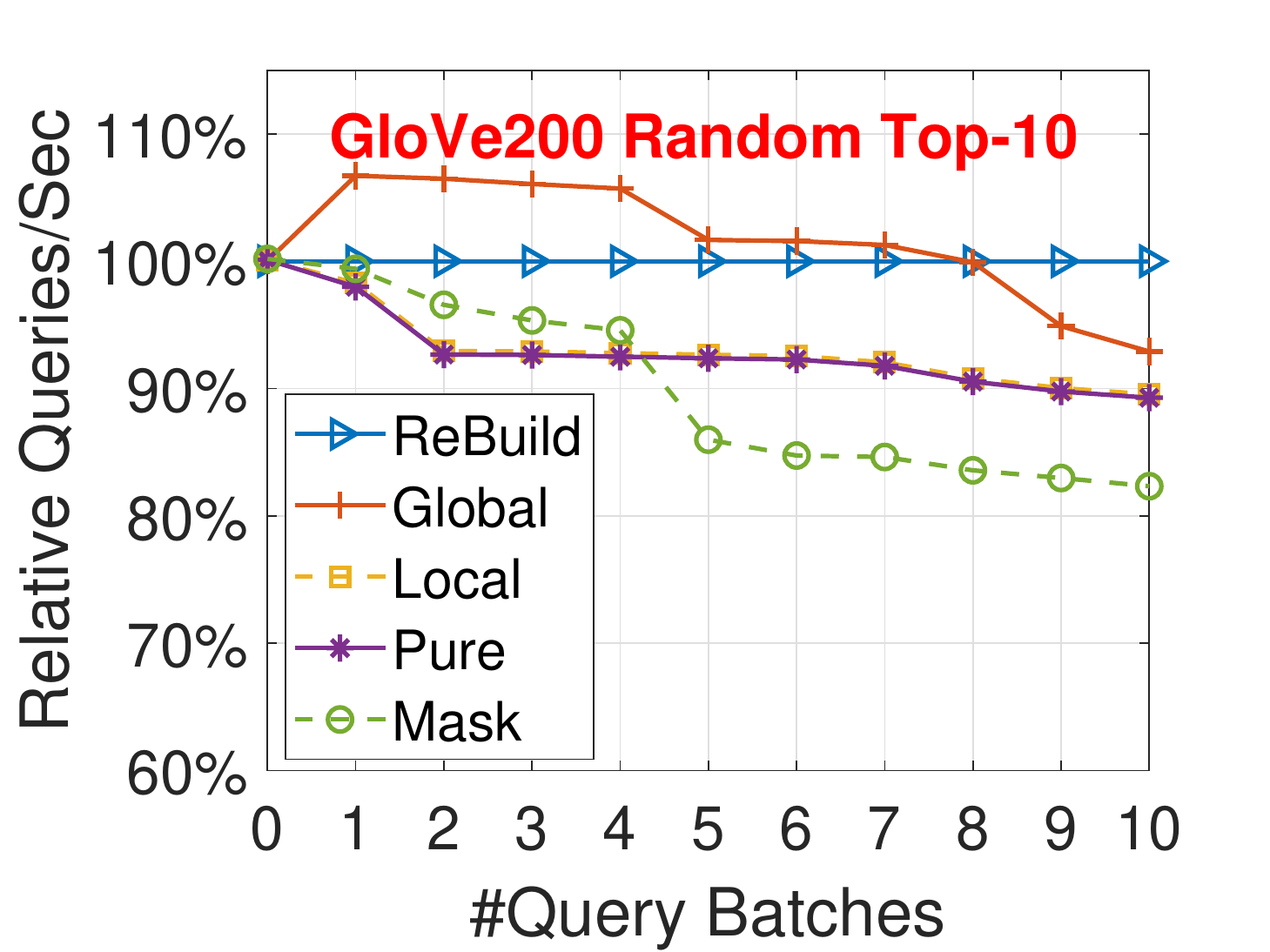}\hspace{-0.1in}
    \includegraphics[width=2.3in]{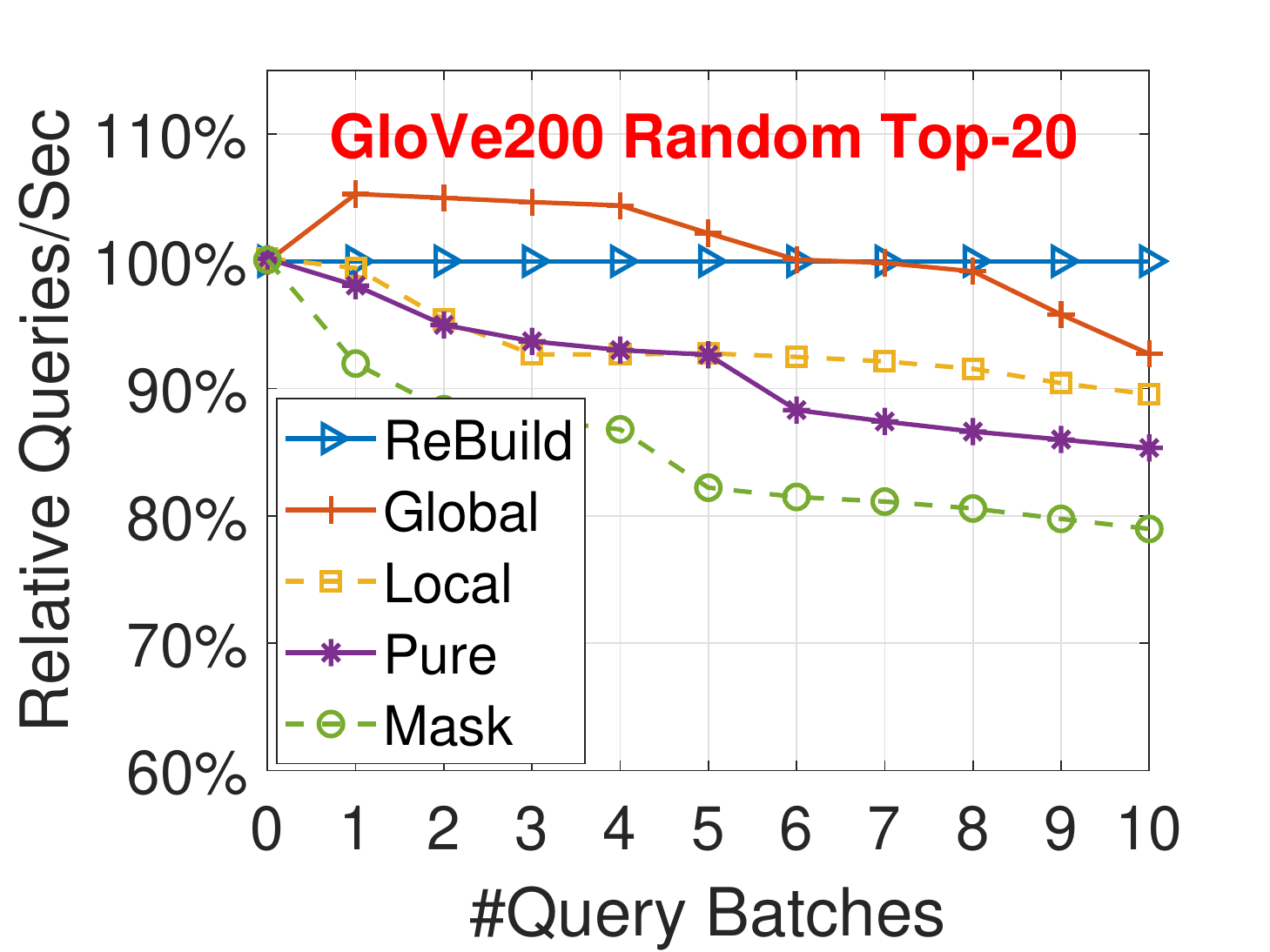}\hspace{-0.1in}
    \includegraphics[width=2.3in]{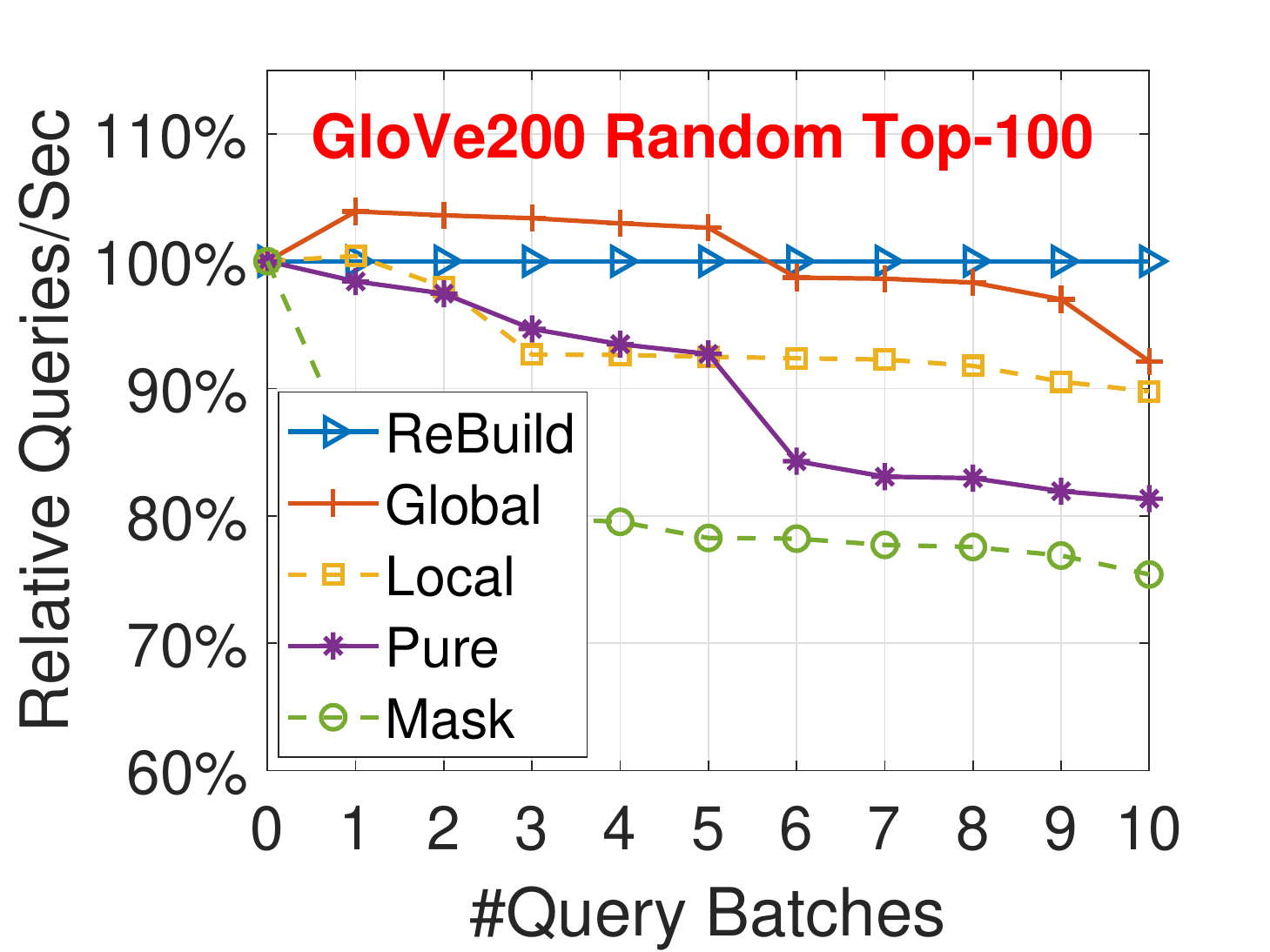}
}

\mbox{\hspace{-0.1in}
   \includegraphics[width=2.3in]{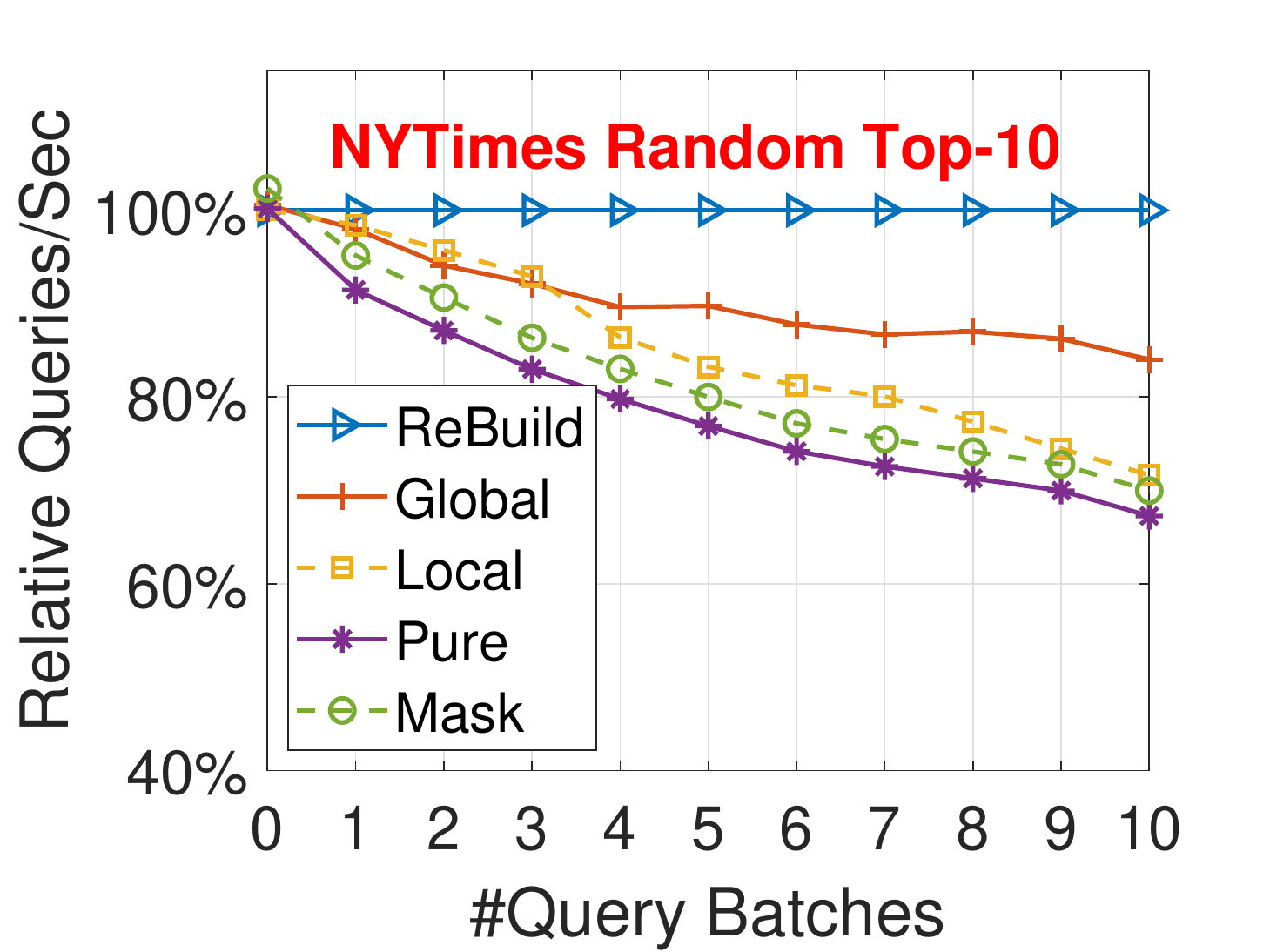}\hspace{-0.1in}
   \includegraphics[width=2.3in]{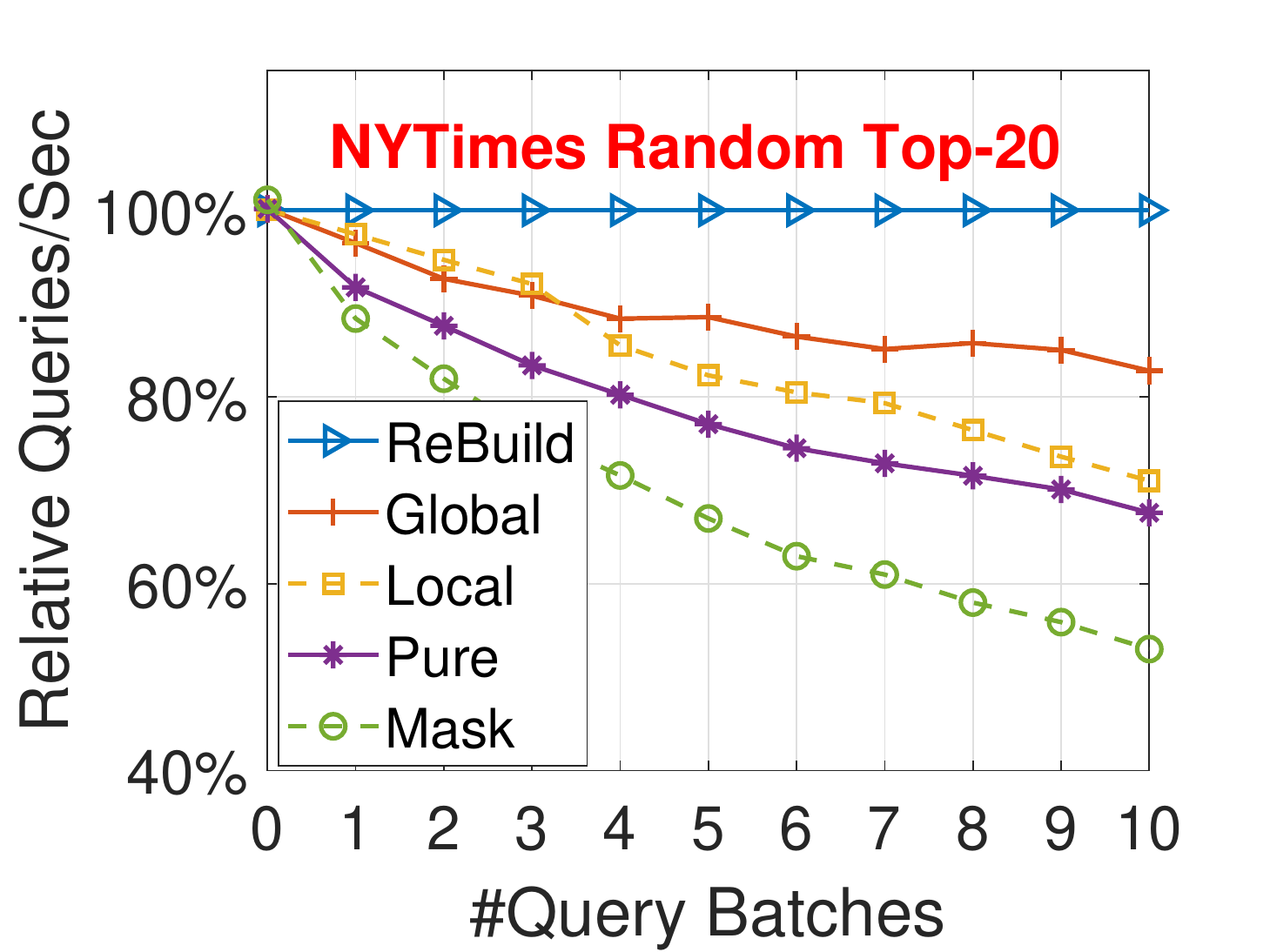}\hspace{-0.1in}
   \includegraphics[width=2.3in]{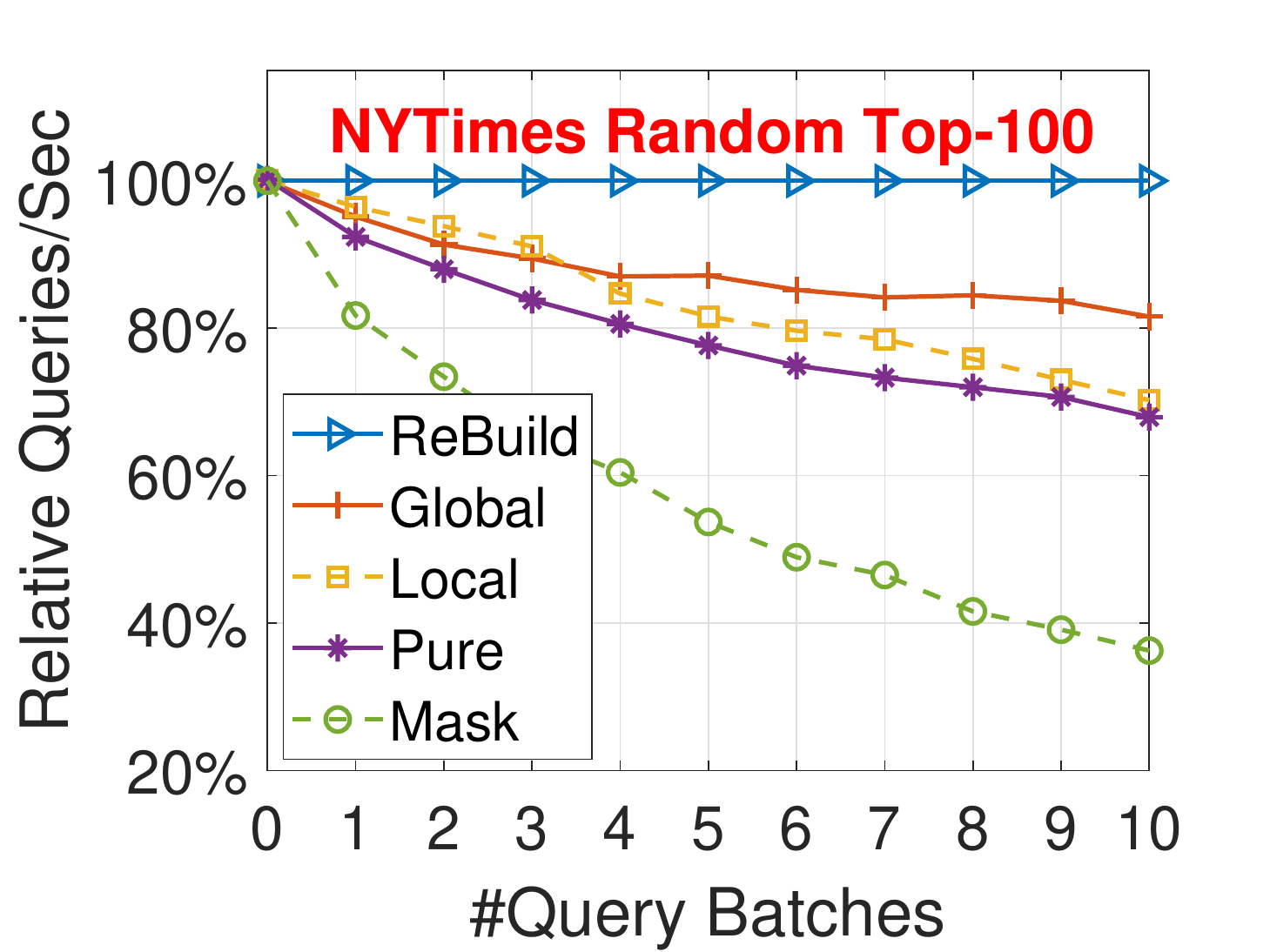}
}

\mbox{\hspace{-0.1in}
   \includegraphics[width=2.3in]{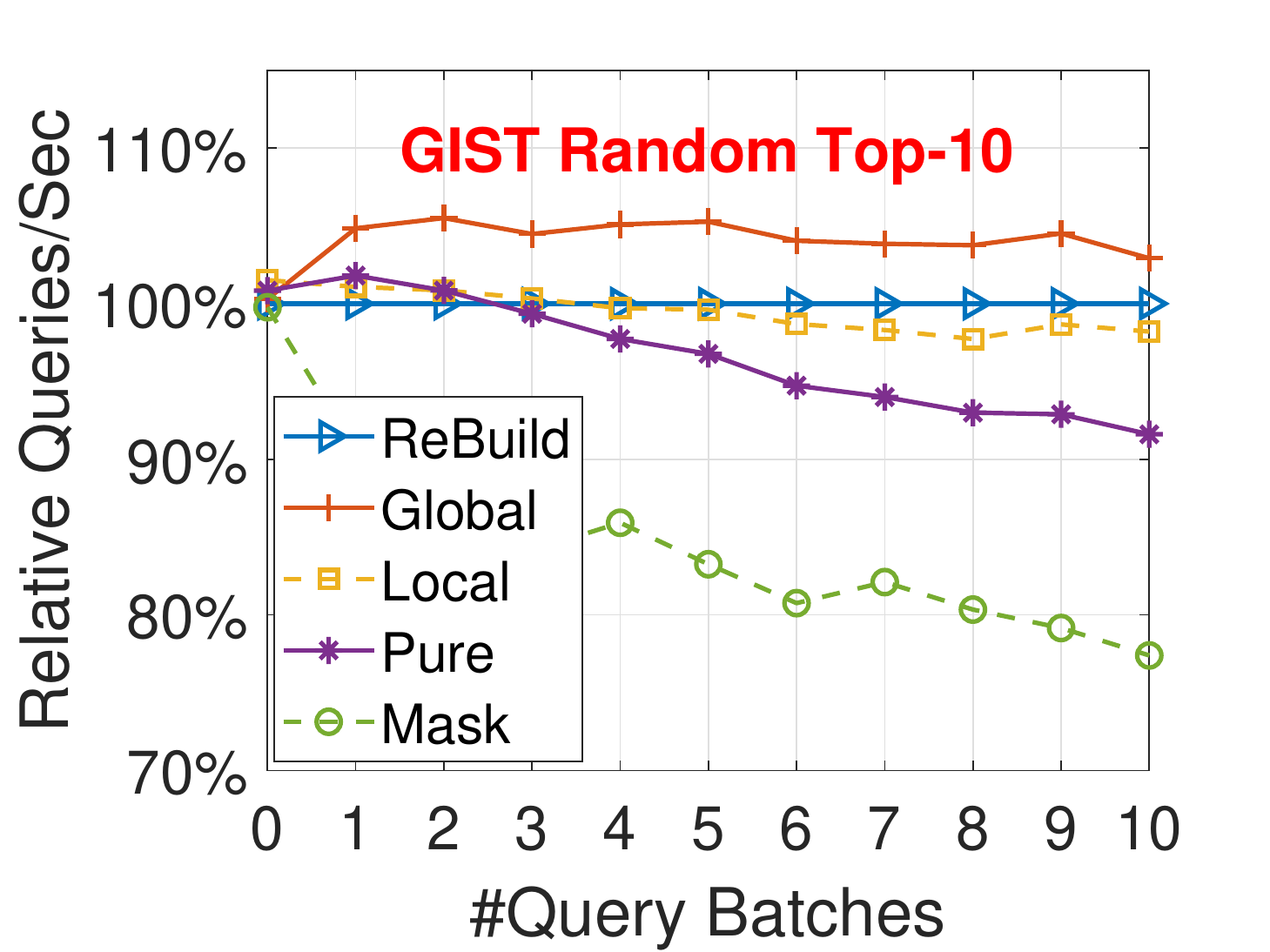}\hspace{-0.1in}
   \includegraphics[width=2.3in]{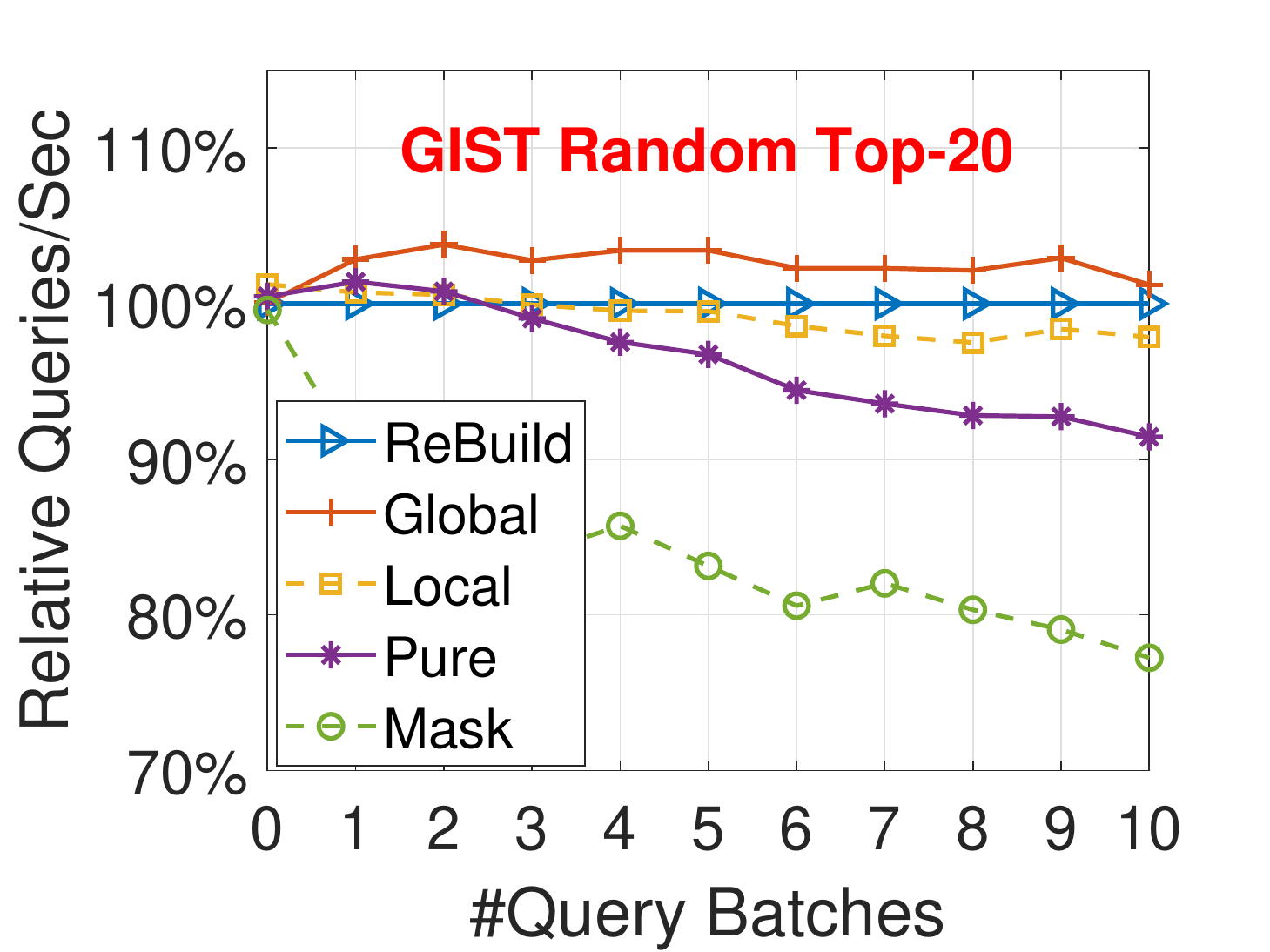}\hspace{-0.1in}
   \includegraphics[width=2.3in]{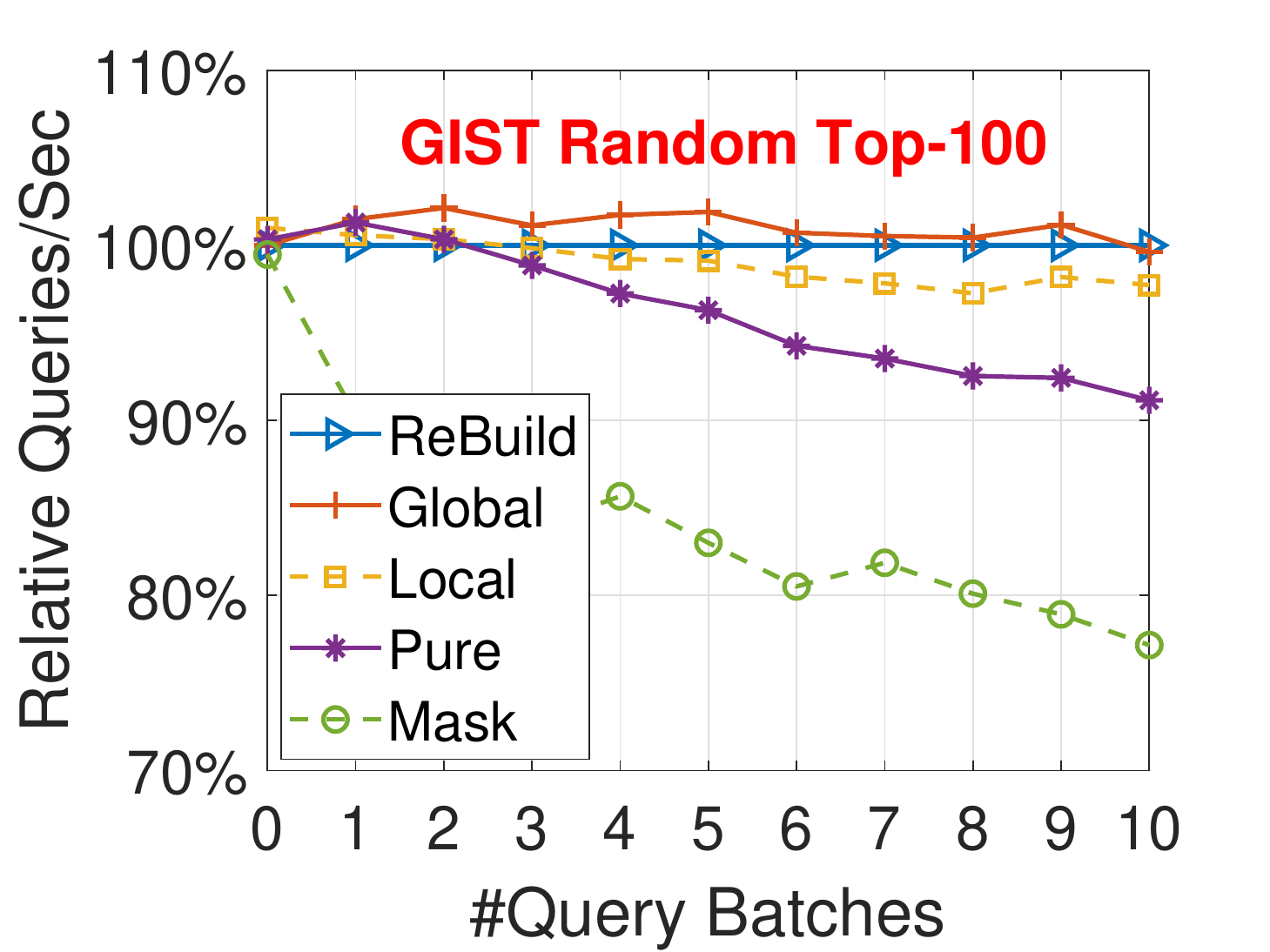}
}

\vspace{-0.2in}

\caption{Relative queries per second to obtain 0.8 recall in each batch. The update pattern for each batch is random updates.}\label{exp:time}
\label{fig:qps_random}\vspace{-0.2in}
\end{figure}

    \subsection{Query Time}
    In this section, we present the query efficiency for our proposed algorithms and baseline. In Figure~\ref{fig:qps_random} and Figure~\ref{fig:cluster_qps}, we plot relative queries processed per second (QPS) versus the query batches of workloads from 3 ANN datasets. The queries processed per second (QPS) versus the query batches of workloads of GIST dataset is also presented at Figure~\ref{fig:cluster_qps}. In each workload, we follow the pattern described above. Here the relative QPS is the result of real QPS divided by the QPS of ReBuild algorithm in 0.8 Top-10, Top-20 or Top-100 recall.

\newpage

\begin{figure}[h!]
\begin{center}
\mbox{\hspace{-0.1in}
    \includegraphics[width=2.3in]{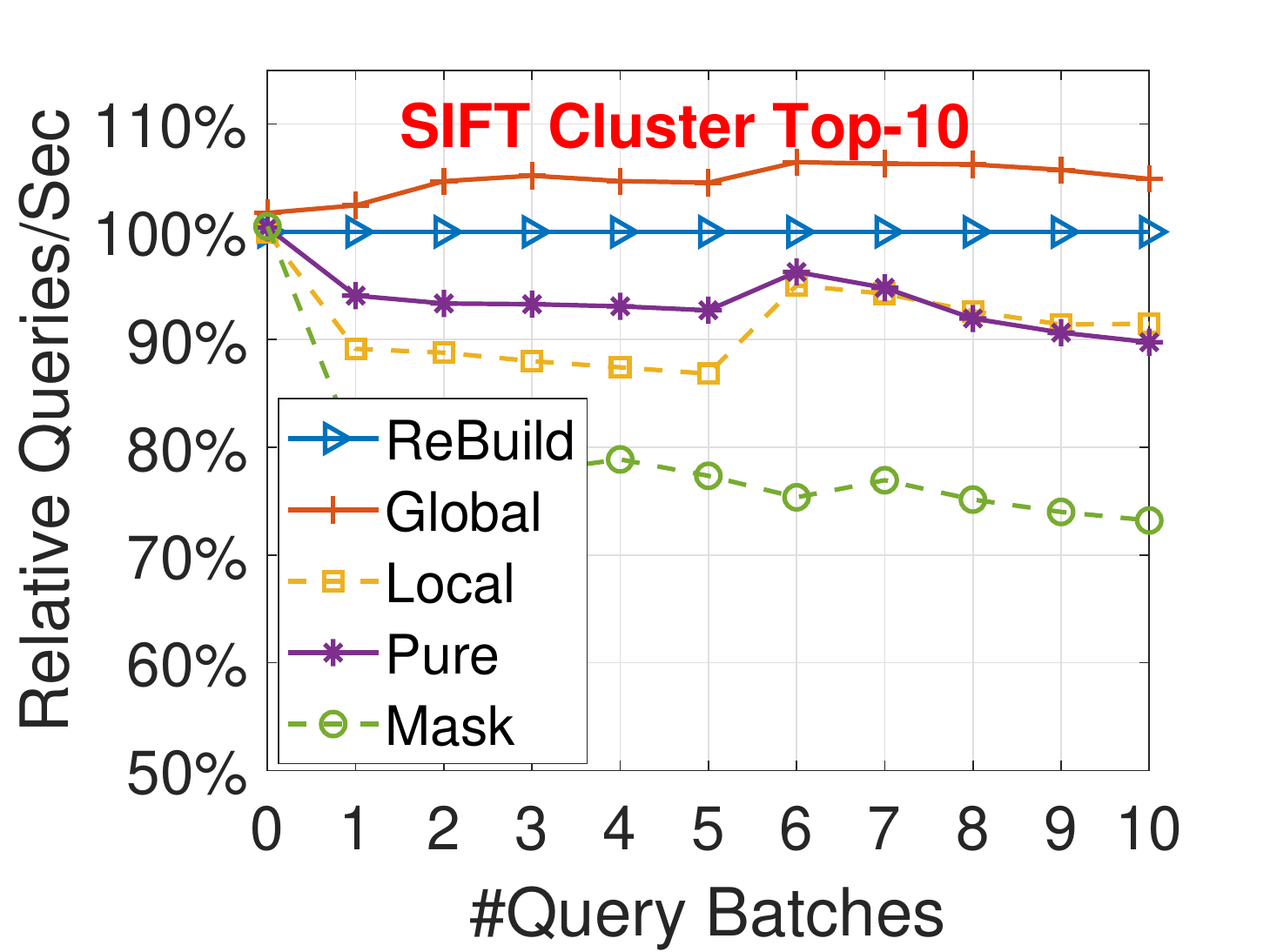}\hspace{-0.1in}
    \includegraphics[width=2.3in]{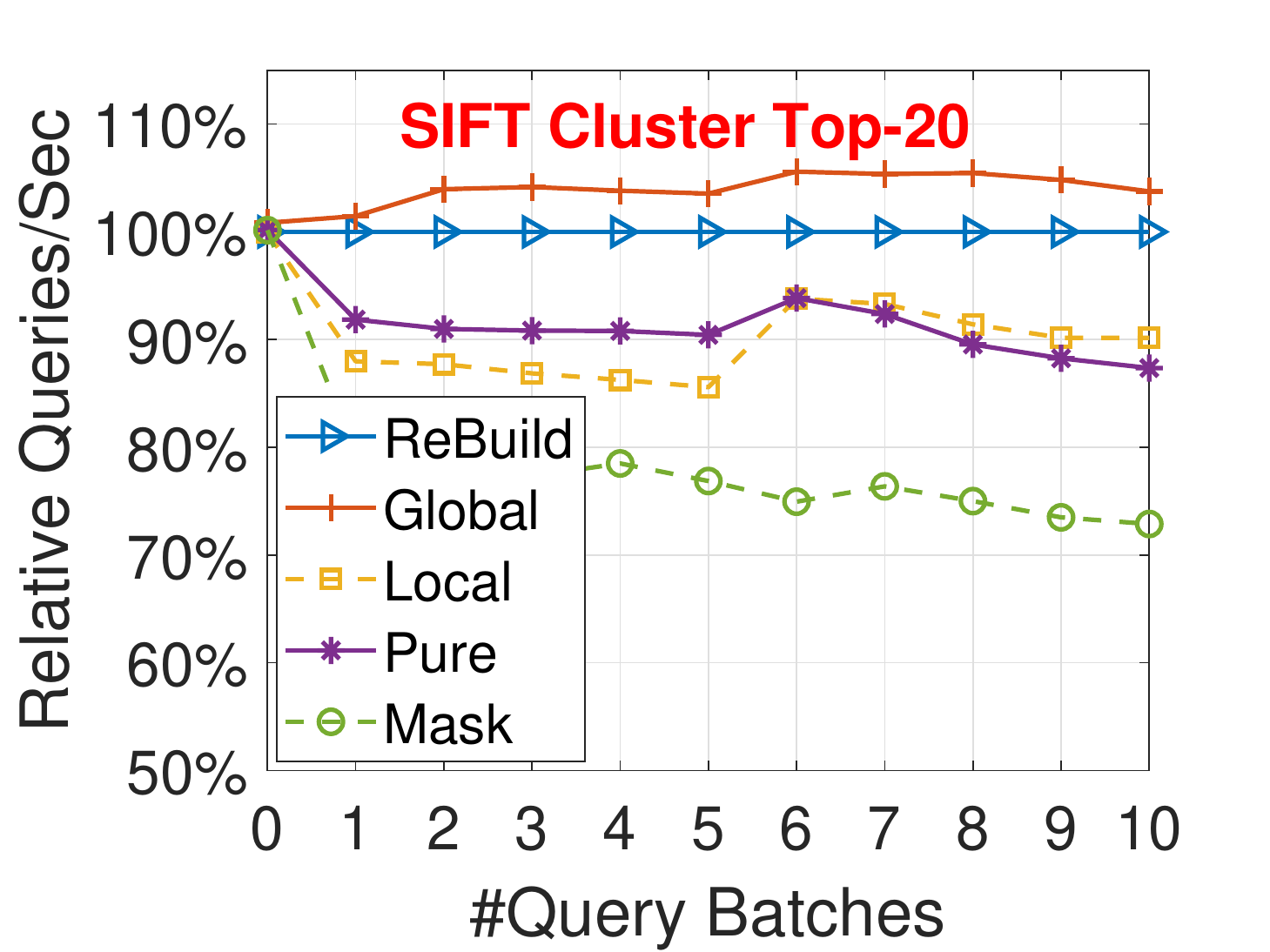}\hspace{-0.1in}
    \includegraphics[width=2.3in]{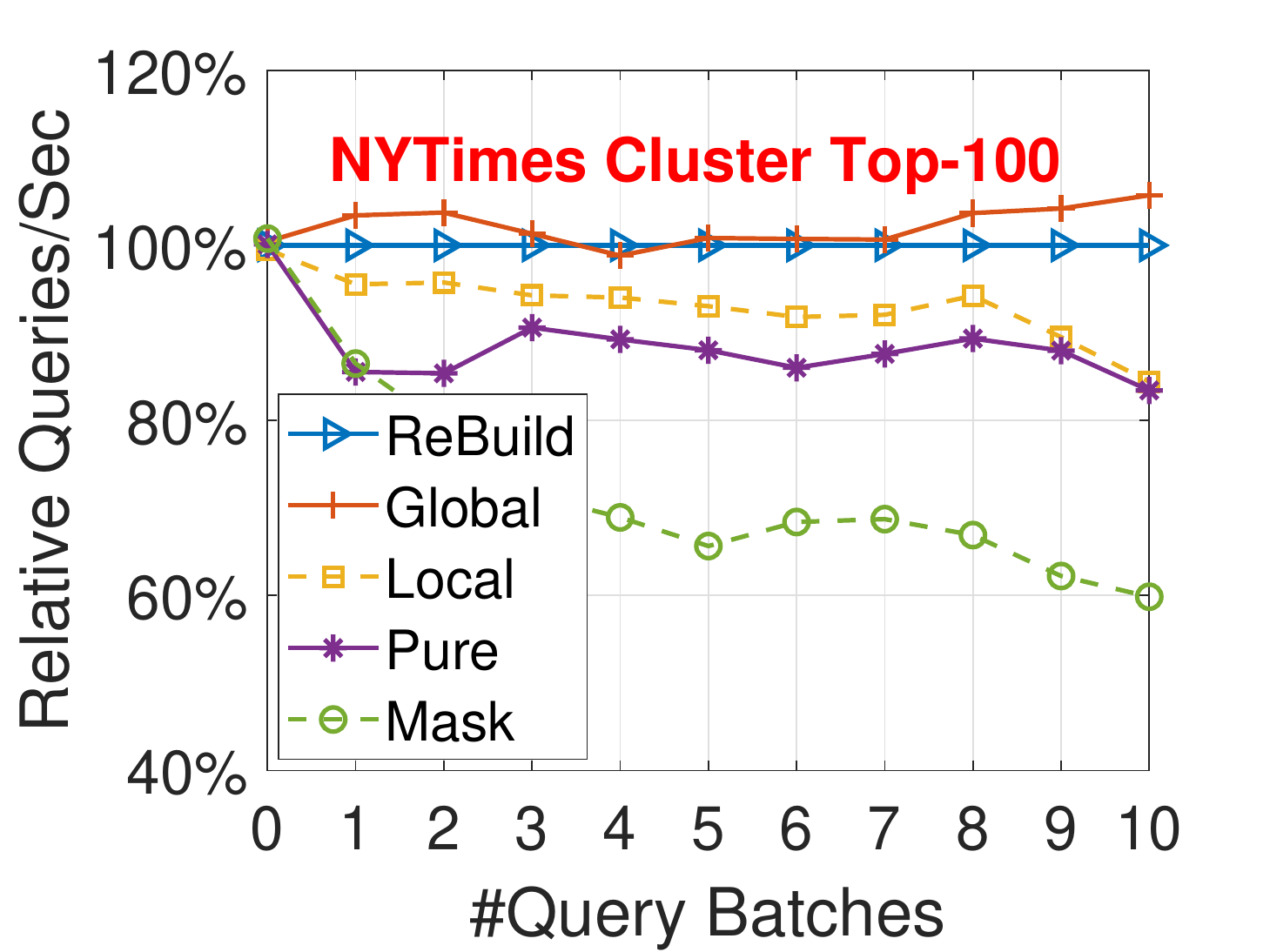}
}
\mbox{\hspace{-0.1in}
    \includegraphics[width=2.3in]{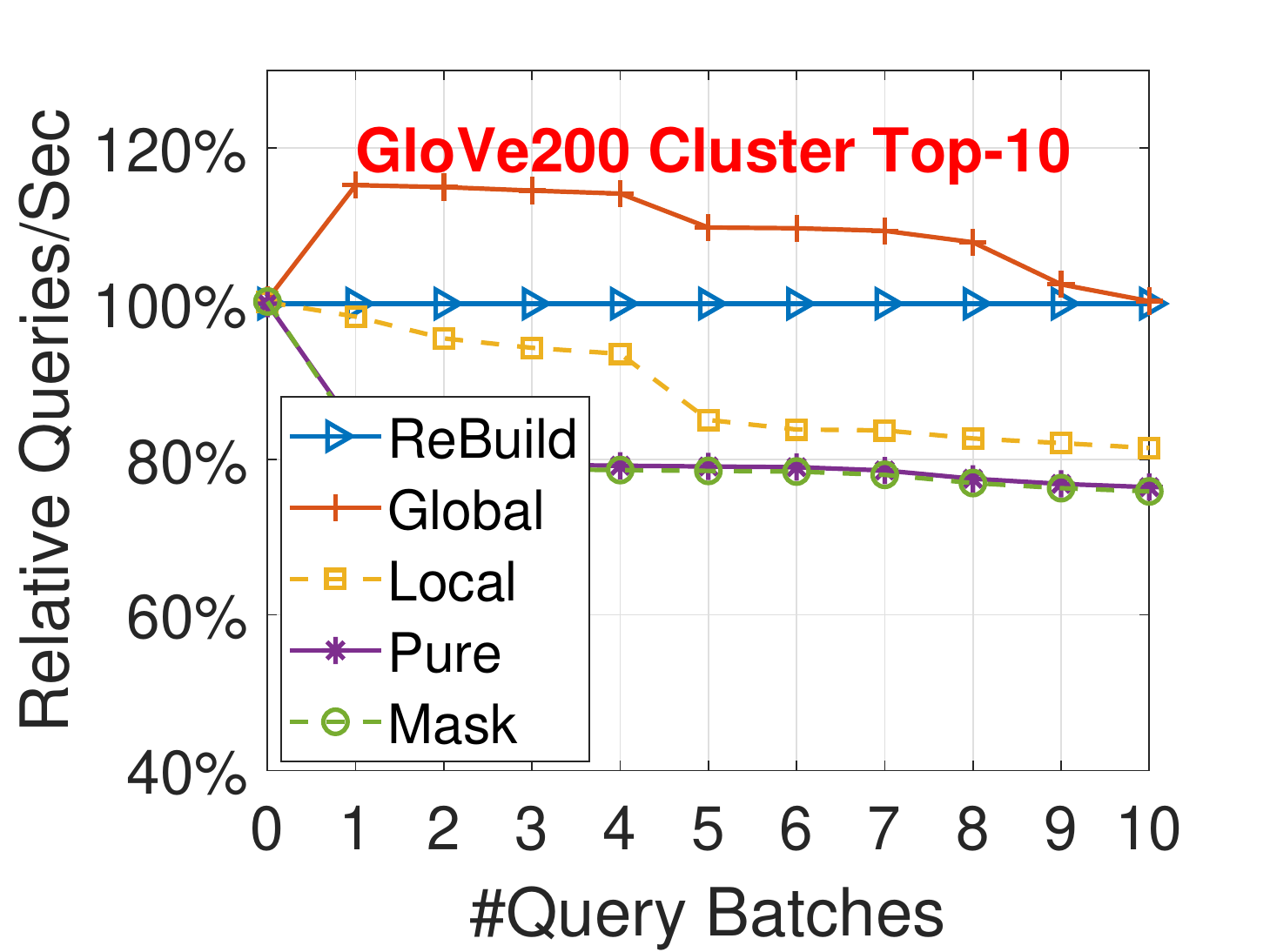}\hspace{-0.1in}
    \includegraphics[width=2.3in]{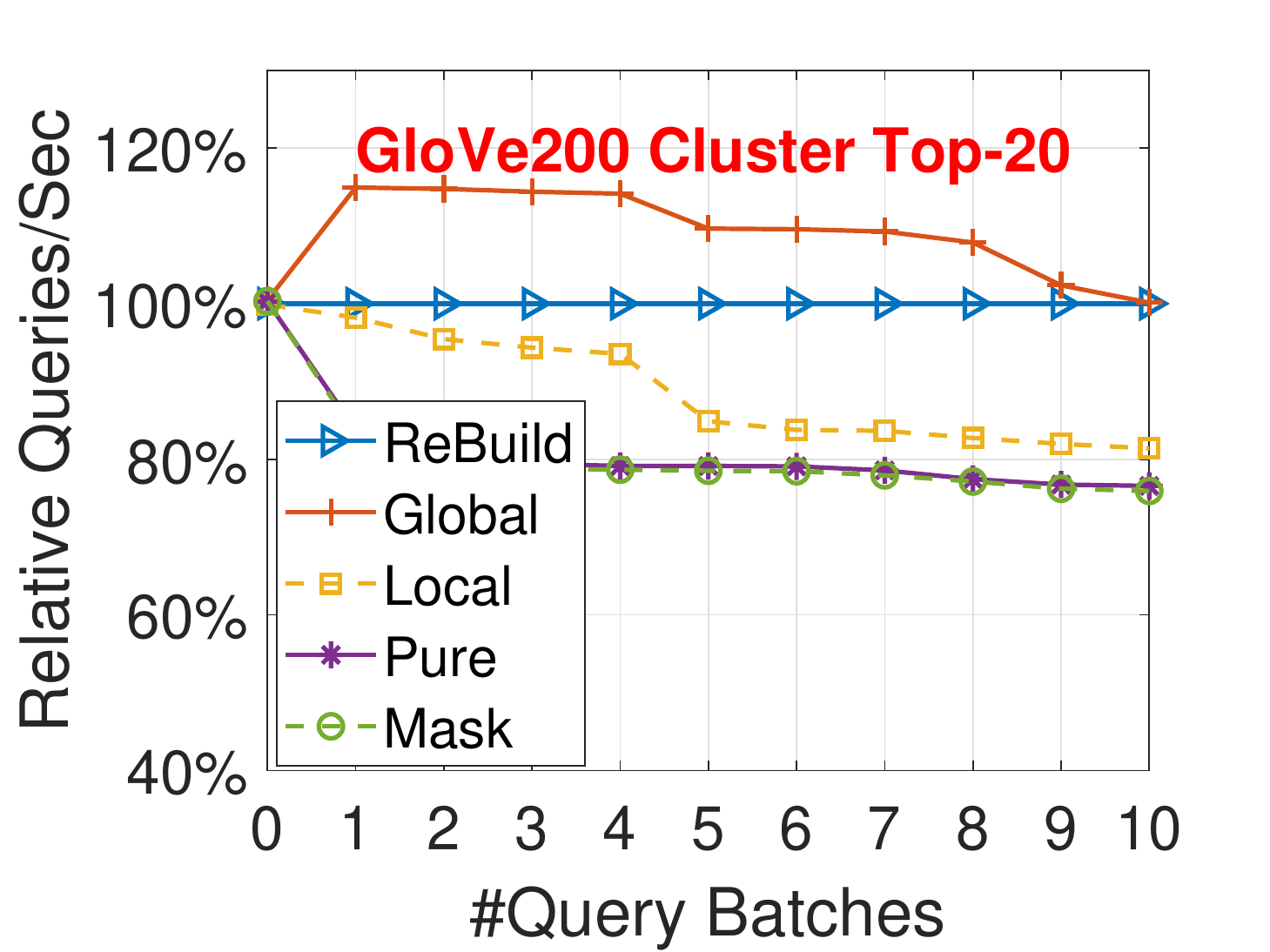}\hspace{-0.1in}
    \includegraphics[width=2.3in]{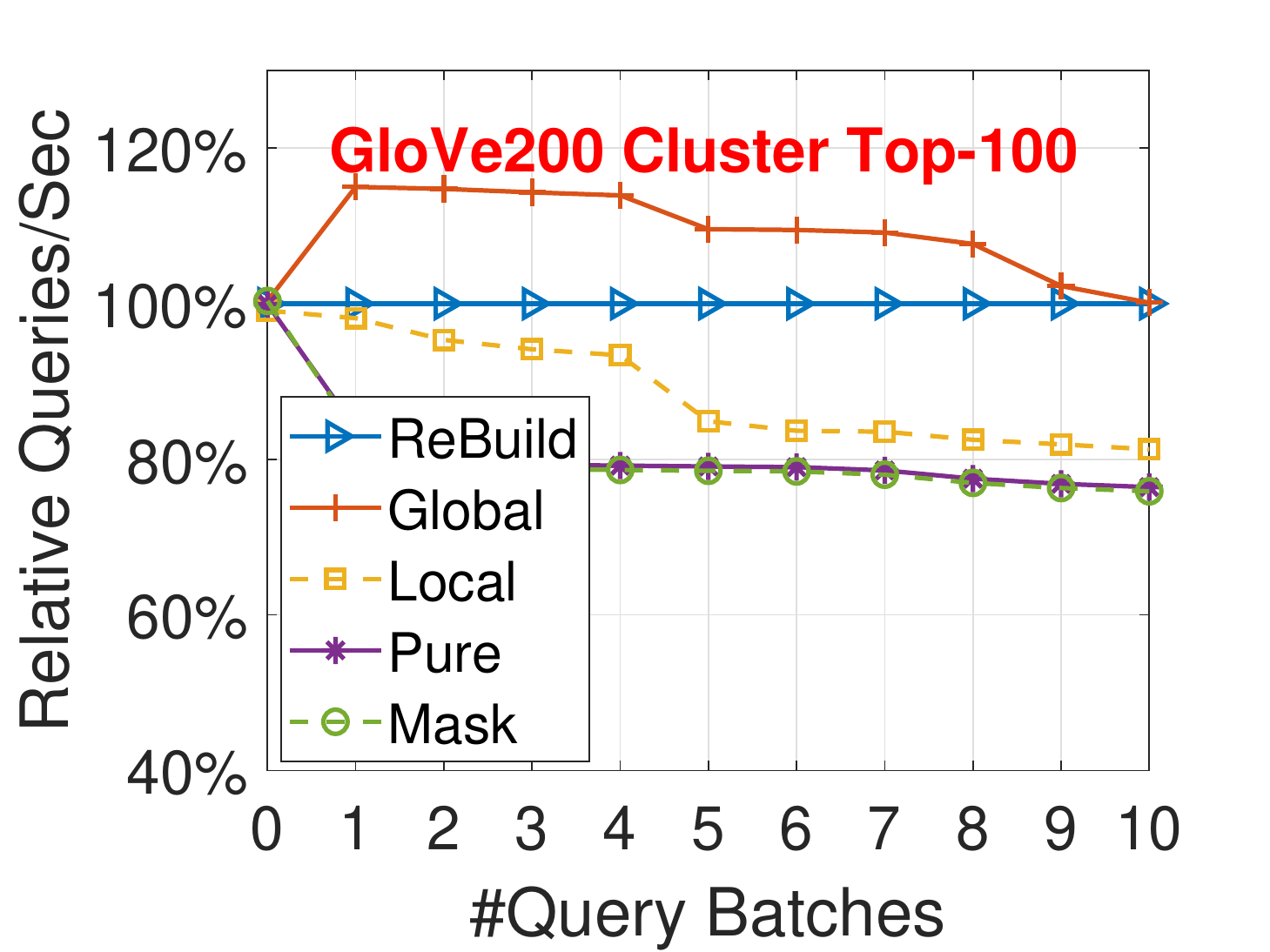}
}

\mbox{\hspace{-0.1in}
        \includegraphics[width=2.3in]{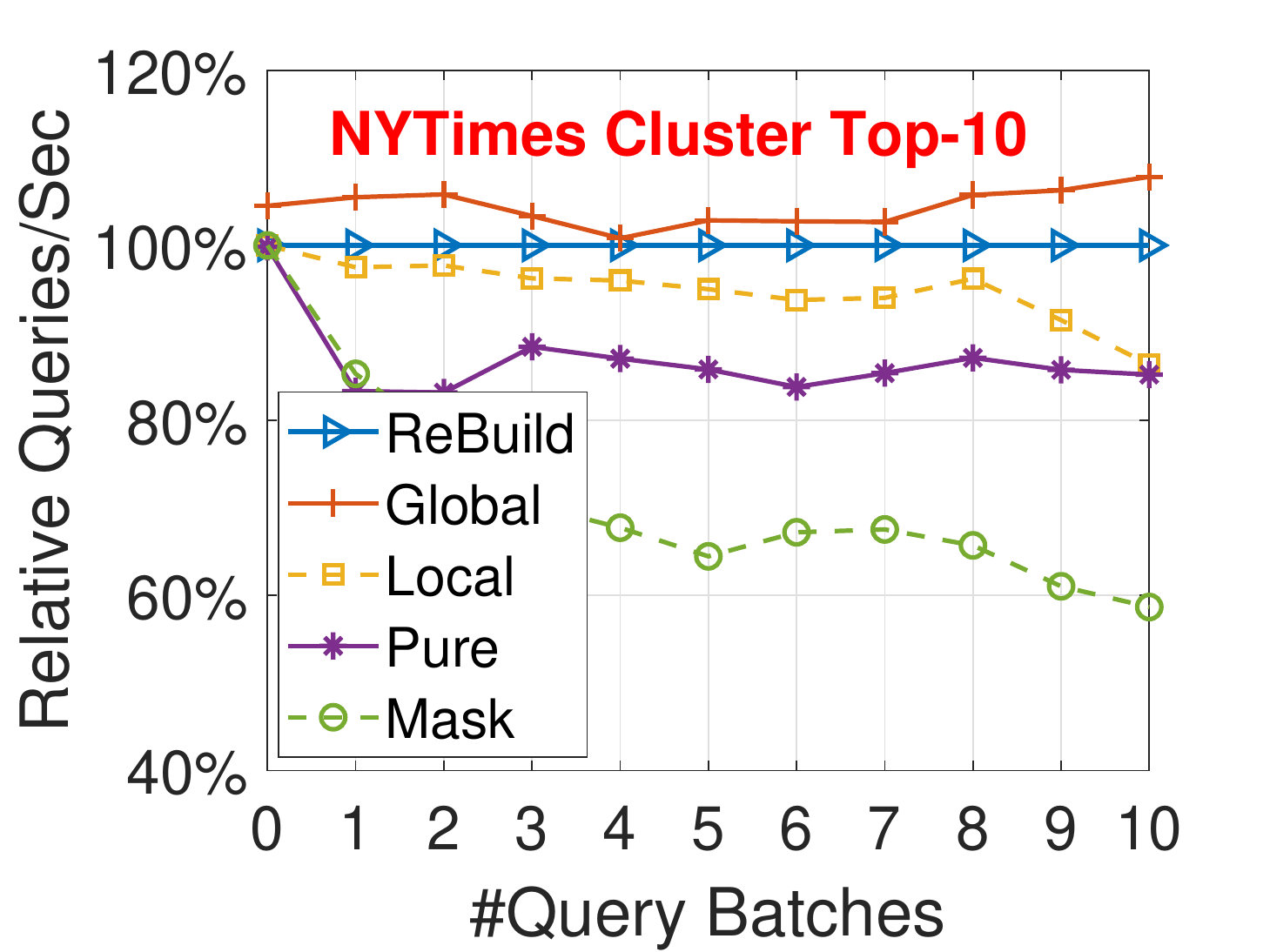}\hspace{-0.1in}
        \includegraphics[width=2.3in]{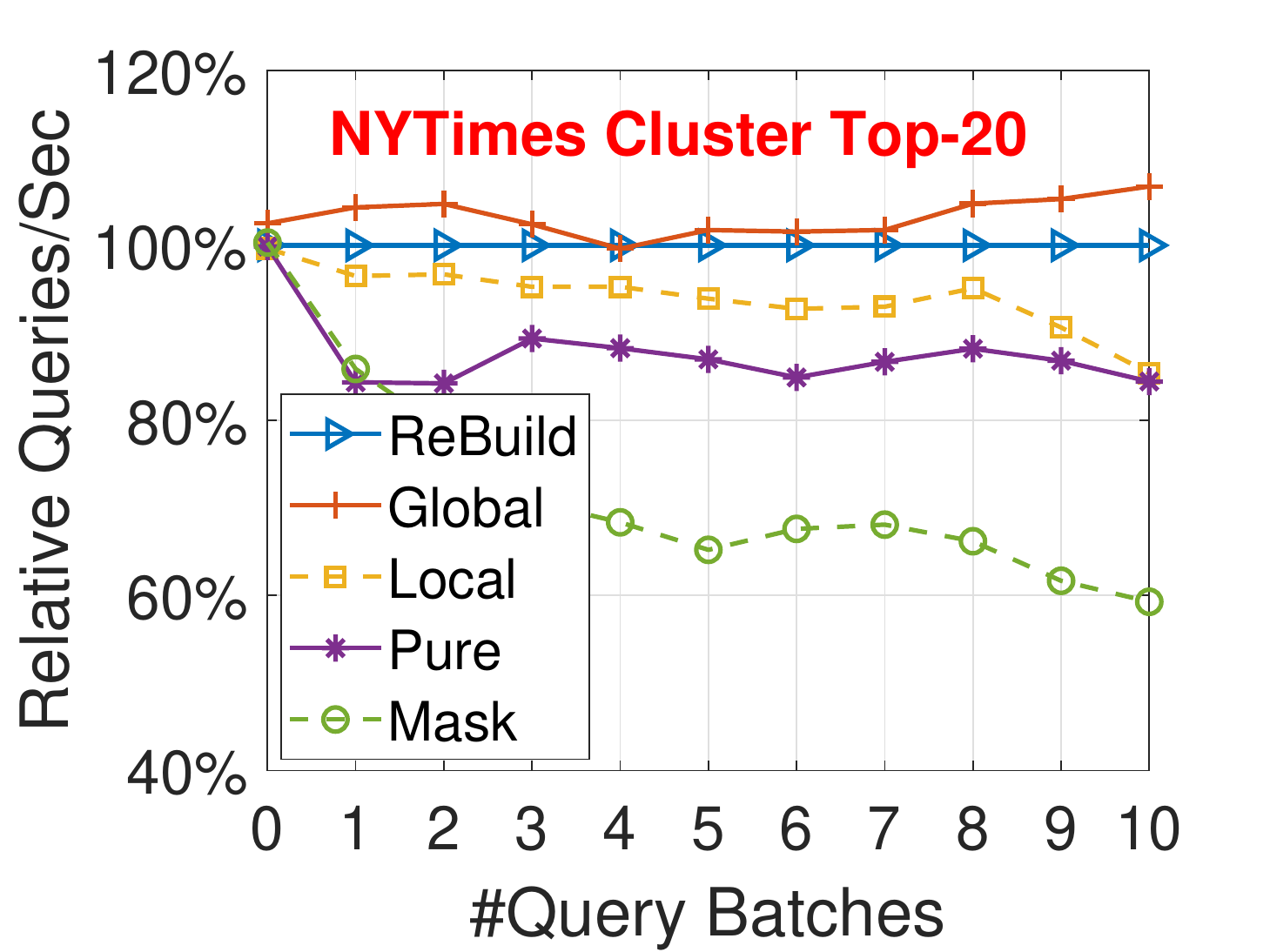}\hspace{-0.1in}
        \includegraphics[width=2.3in]{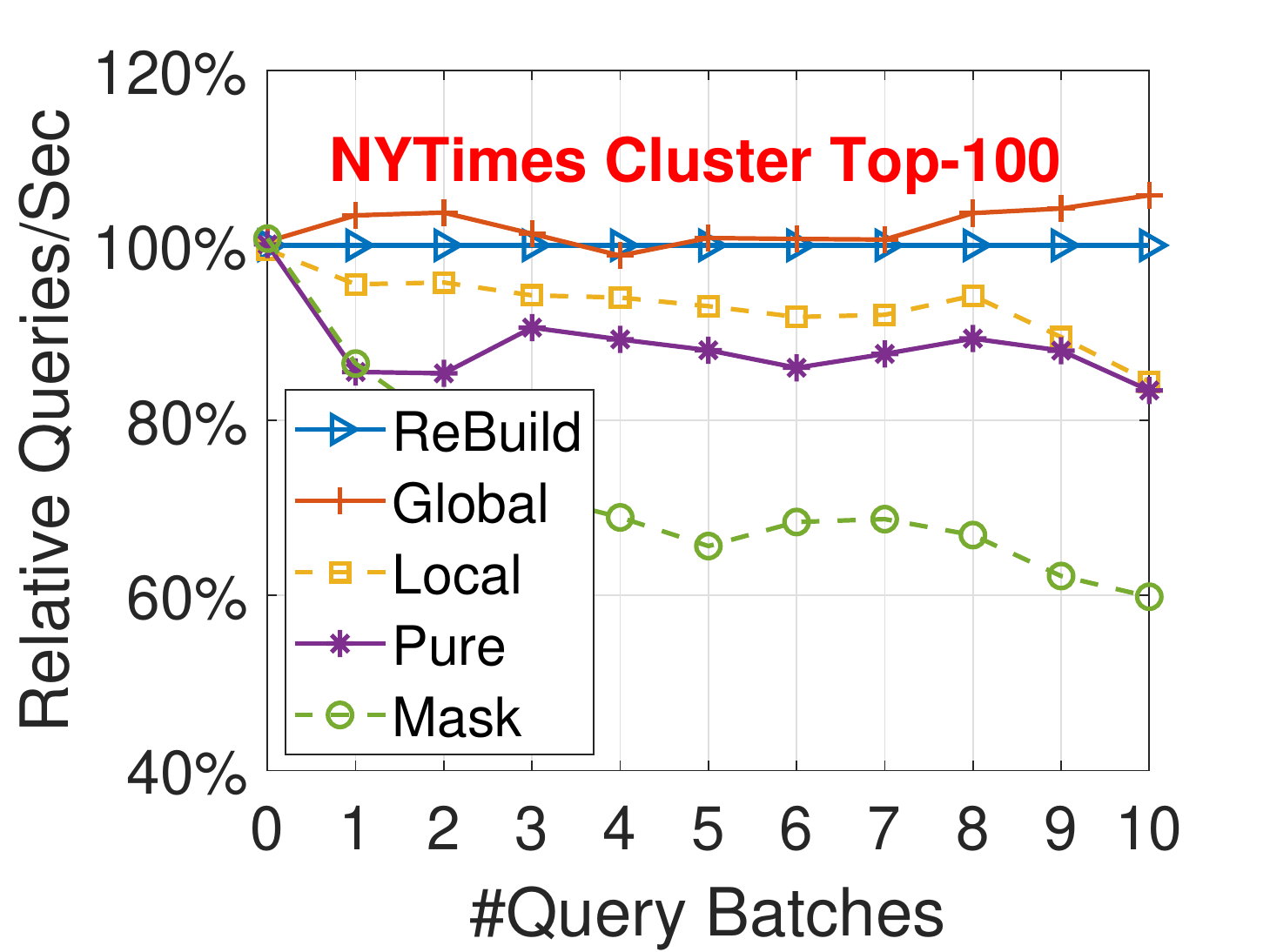}
}

\mbox{\hspace{-0.1in}
        \includegraphics[width=2.3in]{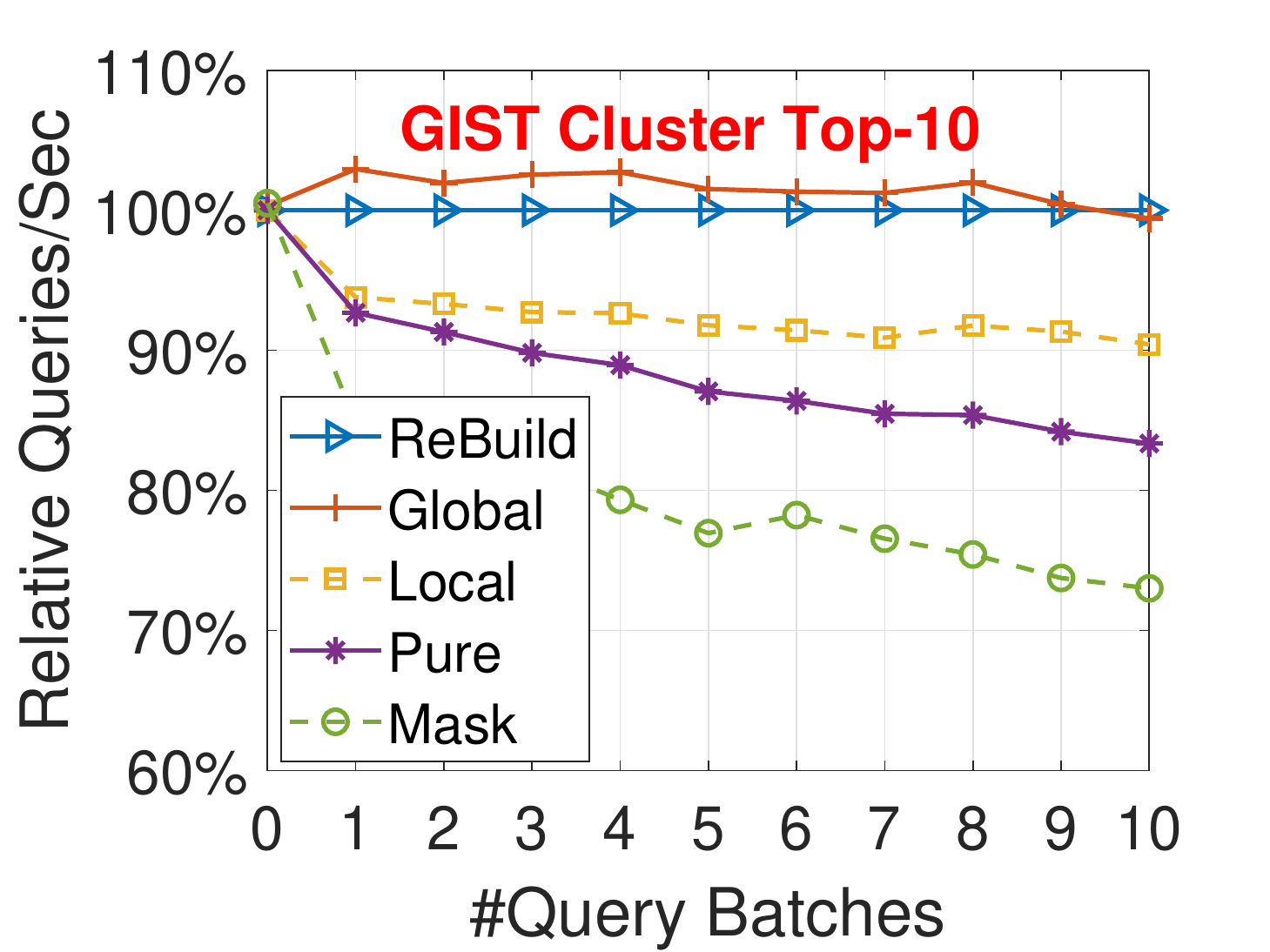}\hspace{-0.1in}
        \includegraphics[width=2.3in]{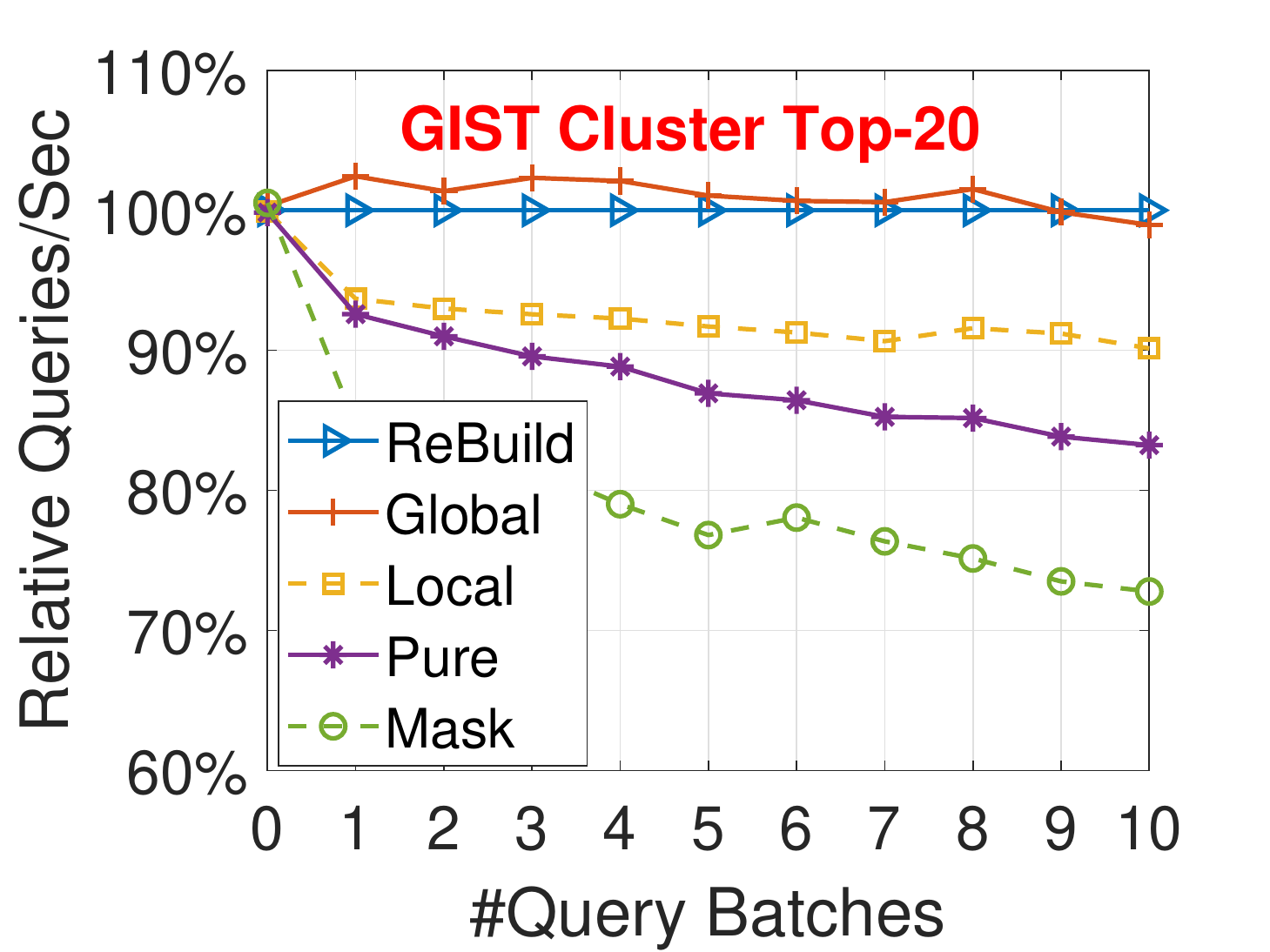}\hspace{-0.1in}
        \includegraphics[width=2.3in]{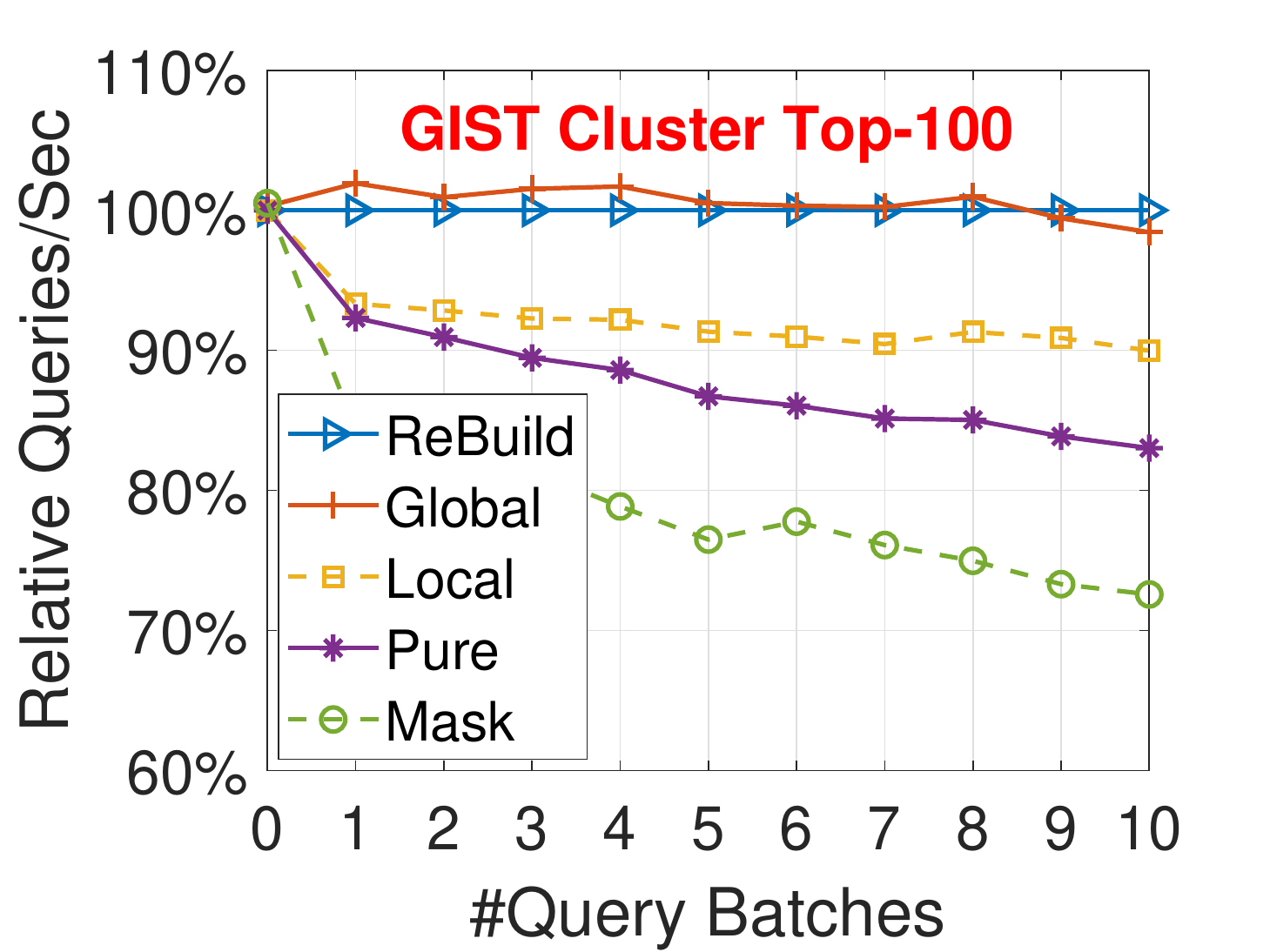}
}

\end{center}

\vspace{-0.2in}

\caption{Relative queries per second to obtain 0.8 recall in each batch. The update pattern for each batch is clustered updates.}\label{fig:cluster_qps} 
\end{figure}

\newpage

    \subsubsection{Random updates}

    The results in Figure~\ref{fig:qps_random} demonstrate that global reconnect achieves superior performance than other delete-update-edge algorithms. For each dataset, the QPSes of different methods start at the same point in query batch 0, which means that they have the same performance when doing ANN search on the base set without an update. Then, as the batch increases, we observe that our global reconnect algorithm has higher QPS than other incremental update methods. In SIFT dataset, the global reconnect algorithm even outperforms reconstruct graph in  QPS at 0.8 Top-10 recall. The same phenomenon also happens in GloVe200 dataset and GIST dataset. For baseline methods, local reconnect achieves the best performance. Pure deletion algorithm has a QPS drop as batch increases, which means the proximity graph becomes inefficiency due to the broken of connectivity. The vertex mask method achieves the lowest QPS in every batch, illustrating that unnecessary visits of expired vertices significantly influence the performance of ANN search on proximity graph.

    \subsubsection{Clustered updates}

    The results in Figure~\ref{fig:cluster_qps} demonstrates the superior performance of global reconnect algorithm in workloads with clustered updates. Similar to Figure~\ref{fig:qps_random}, global reconnect achieves the best QPS in each batch at 0.8 Top-10, Top-20 or Top-100 recall. In the clustered updates settings, we delete some vectors and their nearest neighbors in each update step. Therefore, for vectors located near to these clusters, a large fraction of their nearest neighbors are also deleted. To update their connections, it requires global search on graph to find new nearest neighbors.

    Different from random updates, in clustered updates, we observe that our global reconnect algorithm outperforms graph reconstruction in almost all scenarios.
    Since the graph is constructed incrementally---new edges of an inserted vertex are connected based on the existing vertices, those edges do not accurately reveal the exact nearest neighbors. Our global reconnect algorithm updates edges based on the entire graph. The edges in the graph are refined during the updates---thus, a better performance than the reconstruction is observed.

	\subsection{Overall Execution Time}
	To evaluate the execution time of our proposed global reconnect algorithm and other baselines, we plot the accumulated time versus the number of operations. The accumulated time grows as the number of operations increases. Lower plot in the figure represents less workload execution time and better performance. In this experiments, we consider two major characteristics of real-world scenarios. First, embedding vectors of different usernames, are a much larger group than the online product embeddings. Therefore, the number of queries may be greater or equal to the number of vectors in database. Second,  one single user may look for products several times within an hour. Thus, one query vector may be executed hundreds of times. Combining this two features, we modify the original workload for SIFT. Given 900k base vectors, in every step of the workload, we remove 10k vectors and insert the same amount of vectors, but query with 200k, 1 million or 20 million vectors. The query set is obtained by duplicating original 10k query set. Plots for these 3 workloads corresponds to sub-figures from left to right in Figure~\ref{fig:workload_time}.

\begin{figure}[t]

    \begin{center}
     \mbox{
    \includegraphics[width=3in]{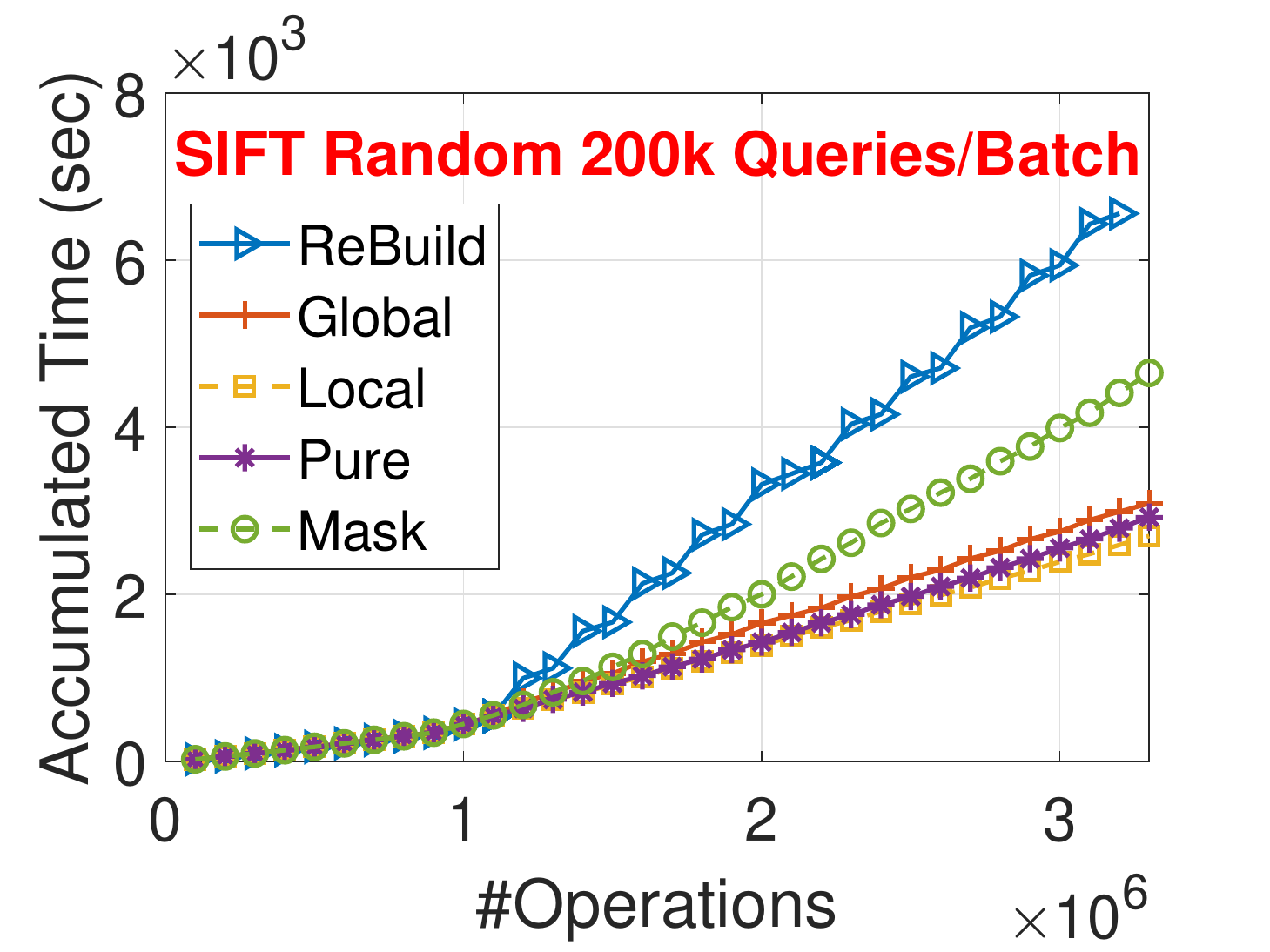}
        \includegraphics[width=3in]{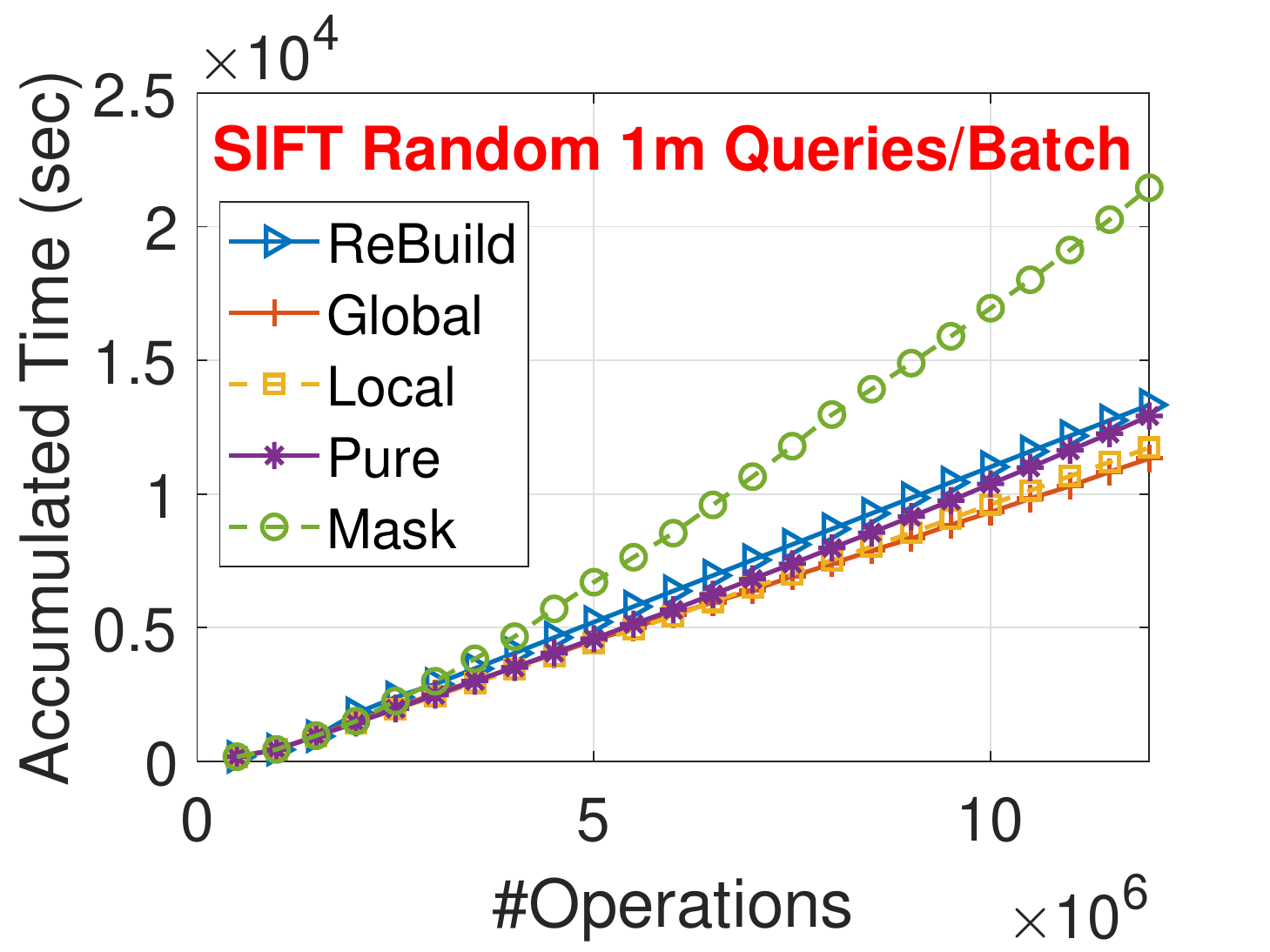}
        }
    \mbox{
     \includegraphics[width=3in]{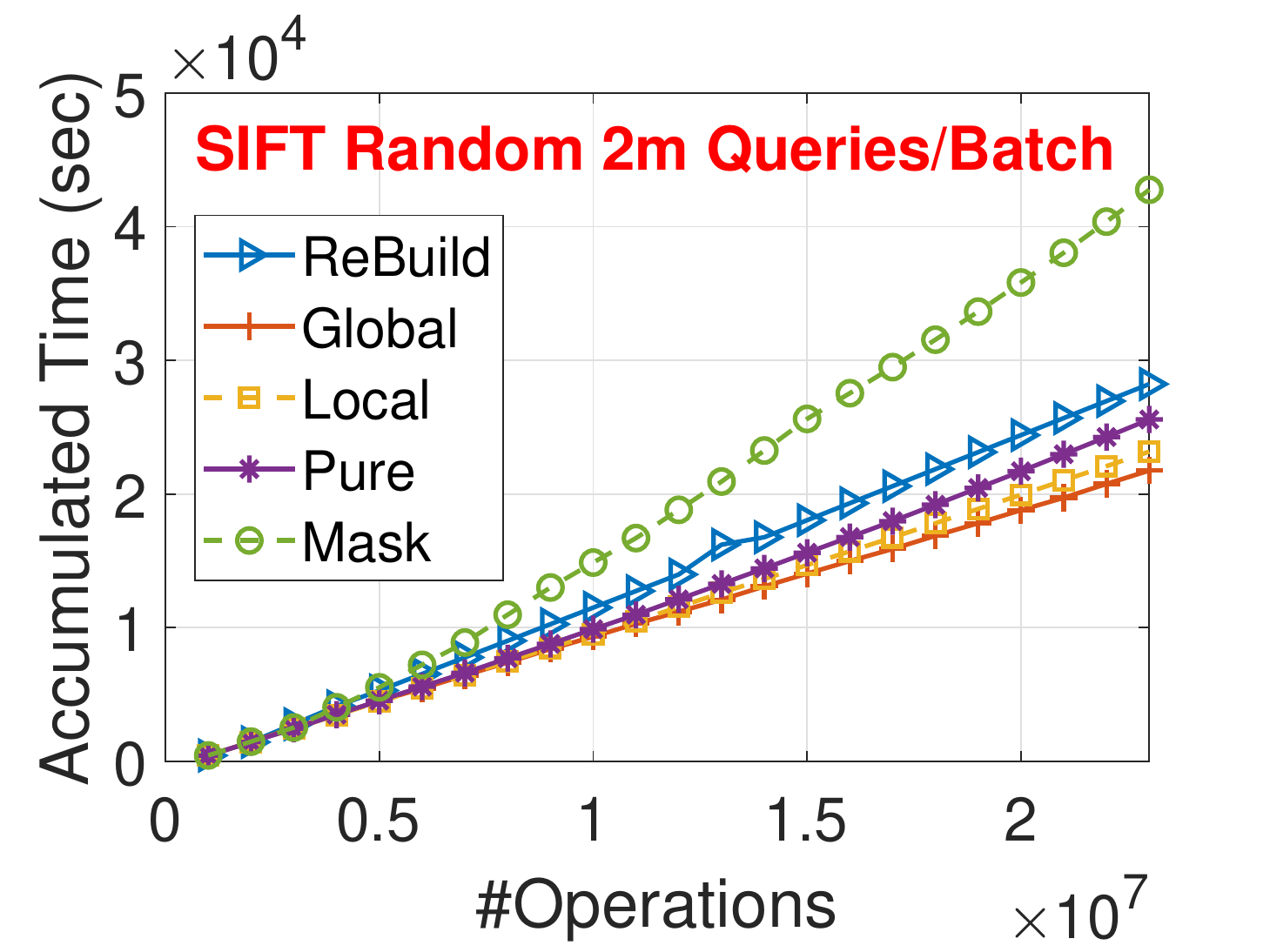}
        \includegraphics[width=3in]{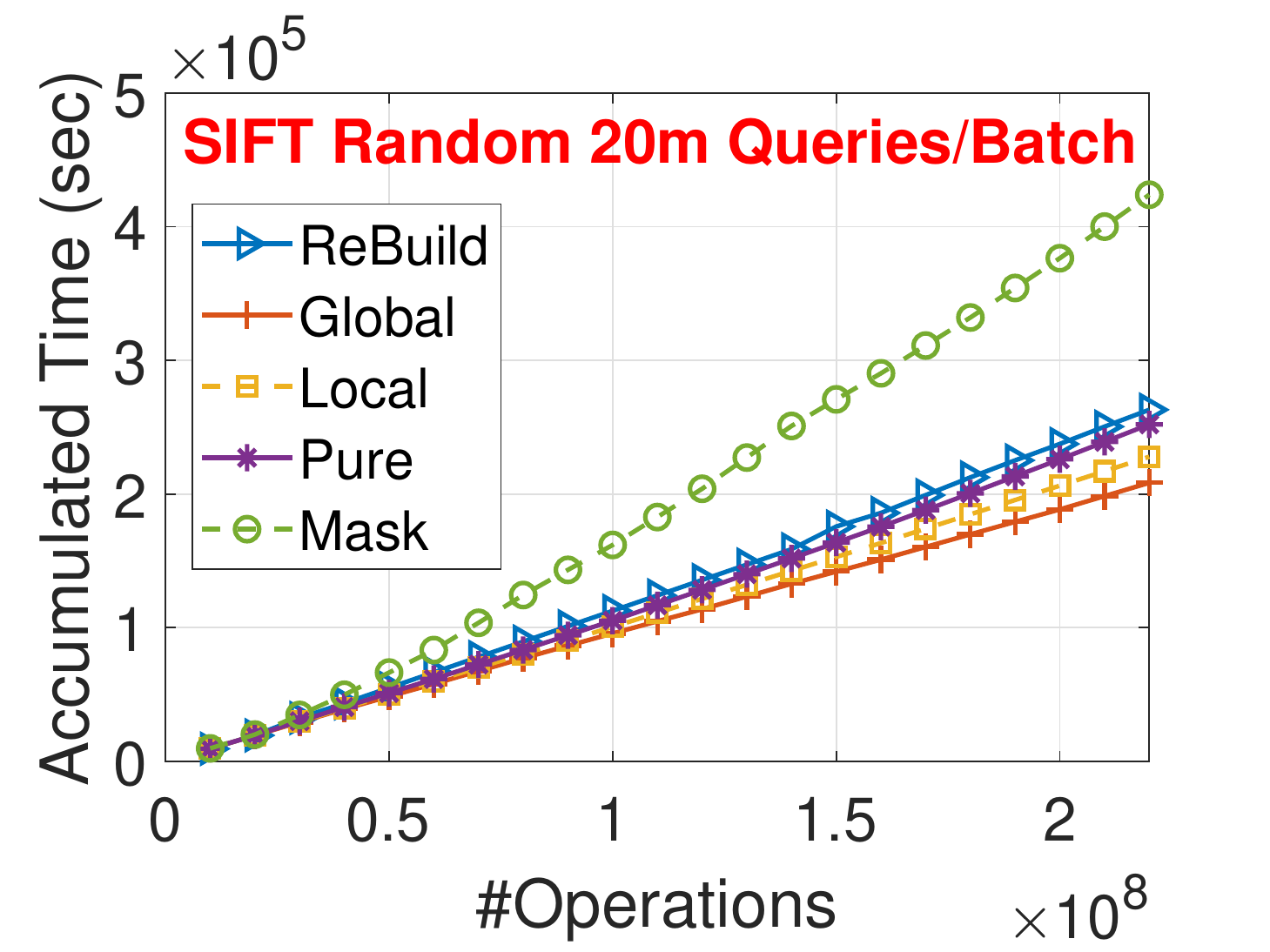}
            }
    \end{center}

    \vspace{-0.15in}

    \caption{Total execution time vs number of operations}\label{fig:workload_time}
    \end{figure}

\newpage

    In Figure~\ref{fig:workload_time} on the left, we show that if we query 200k vectors per batch. Graph reconstruction has the highest time cost because the time cost introduced by reset all vertices on graph. Then, the vertex masking is also time inefficient. Although it deletes vertices in the fastest way, the significant QPS drop in query phase makes it time-consuming. The global reconnect, local reconnect and pure deletion have similar total execution time. Global reconnect and local reconnect have higher execution time due to the cost introduced by reconnection. Then, as the number of queries per batch increases, we observe that our global reconnect algorithm achieves the lowest total execution~time.

    This phenomenon validates the assumption that the updating time cost can be amortized by the time saved by efficient ANN search. The same compensation effect is also observed in graph reconstruction and local reconnect. As the number of queries increases each batch, the total execution time of graph reconstruction outperforms vertex masking and becomes closer to pure deletion. The performance of local reconnect also improves as the margin between it and pure deletion increases when we have more queries. Therefore, in real-world applications that require online ANN search. Given the large amount of query requests, the proposed IPGM framework outperforms static proximity graph based methods due to the advantages taken in each search phase. Meanwhile, IPGM reduces the latency of data deletion by avoiding reindexing the whole dataset as a graph~in~each~step.

    \subsection{Further Discussions}
     \begin{itemize}[leftmargin=*]
      \setlength\itemsep{0.1em}
    \item  To obtain a level of recall at 0.8, the proposed global reconnect algorithm improves the performance of incremental proximity graph in online ANN search by a factor of 10.6$\%$ on average. In some scenarios, global reconnect algorithm also outperforms graph reconstruction by a maximum factor of 18.8$\%$.

    \item The proposed global reconnect algorithm is robust to random update pattern and clustered update pattern. In clustered update when we delete a vector and its nearest neighbors together, the global reconnect algorithm is a better way to maintain a good approximation to Delaunay Graph.

    \item Our proposed global reconnect algorithm reduces the total execution time when the number of queries increases. In real-world settings where the number of queries and data vectors are at the same level, our method has the fastest execution time compared to other methods.
    \end{itemize}

\section{Conclusion}

    Graph-based ANN search methods index the dataset as proximity graph and has shown superiority over other ANN search approaches. However, the inefficiency of proximity graph against incremental data vector insertion and deletion prevents it to be used in large-scale online ANN search. This paper presents a proximity graph maintenance algorithm for online ANN search. We define the incremental proximity graph update problem formally and give theoretical proofs to demonstrate the feasibility of incremental vertex insertion and deletion on the proximity graph. We propose 4 proximity graph update algorithms to preserve the advantages of proximity graph in ANN search when vertices are removed from the graph in different patterns. The experiments on 4 public ANN datasets show that our proposed algorithms outperform the reconstruction baseline by a maximum factor of 18.8$\%$. The proposed updating scheme eliminates the performance drop in online ANN methods on proximity graphs, which enables the deployment in real-world search engines and recommendation systems.

	\newpage
	
	\bibliographystyle{plainnat}
	\bibliography{custom,refs}

\end{document}